\documentclass[aps,twocolumn,superscriptaddress,nofootinbib,10pt,pra]{revtex4-1}%put pra here so that we get section numbers in the appendix
\usepackage{paralist} %for compactenum
\usepackage[utf8]{inputenc}
\usepackage[english]{babel}
\usepackage[T1]{fontenc} %for what? That's why: http://tex.stackexchange.com/questions/664/why-should-i-use-usepackaget1fontenc
\usepackage{amsmath,amsfonts,amssymb,amsthm,graphicx,bbm,enumitem,times}
\usepackage{xcolor}
\usepackage{mathtools}
\usepackage{graphicx} % For images
\usepackage{float}    % For tables and other floats
\usepackage{verbatim} % For comments and other
\usepackage{amsmath}  % For math
\usepackage{amssymb}  % For more math
\usepackage{listings} % For source code
\usepackage{simplewick} %For Wick
\usepackage[pdftex]{hyperref}           % For hyperlinks and indexing the PDF

\newtheorem{dfn}{Definition}
\newtheorem{thm}{Theorem}
\newtheorem*{thm*}{Theorem}
\newtheorem{lm}[thm]{Lemma}

\DeclareMathOperator{\dist}{d}
\newcommand{\R}{\mathbb{R}}
\newcommand{\N}{\mathbb{N}}
\newcommand{\C}{\mathbb{C}}

\DeclareMathOperator{\tr}{tr}
\DeclareMathOperator{\supp}{supp}

\newcommand{\ket}[1]{\left \vert #1 \right \rangle}
\newcommand{\bra}[1]{\left \langle #1 \right \vert}

\newcommand{\lat}{\mathcal{L}}
\newcommand{\vol}{V}

\renewcommand{\L}{L}
\newcommand{\dlat}{d_\mathcal{L}}
\newcommand{\Cclust}{C_{\mathrm{Clust}}}
\newcommand{\Chat}{\hat{P}}

\newcommand{\Chom}{C_{\mathrm{Trans}}}
\newcommand{\exphom}{\alpha_{\mathrm{Trans}}}
\newcommand{\trec}{t_{\mathrm{Rec}}}
\newcommand{\trelax}{t_{\mathrm{Relax}}}

\newcommand{\cone}{\mathcal{C}}

\newcommand{\id}{\mathbbm{1}}
\newcommand{\expv}[1]{\langle#1\rangle_{\rho}}
\newcommand{\expvG}[1]{\langle#1\rangle_{\rho_G}}
\newcommand{\expvS}[1]{\langle#1\rangle_{\sigma}}
\newcommand{\iu}{\mathrm{i}}
\newcommand{\e}{\mathrm{e}}

\DeclareMathOperator{\Pf}{Pf}
\DeclareMathOperator{\sign}{sign}
\DeclareMathOperator*{\argmax}{arg\,max}

\newcommand{\cstr}[2]{\mathcal{K}^{#1}_{#2}}		% contraint for the indices {set they live in } {partition}
\newcommand{\partitions}{\mathcal{P}}		% set of all partitions [base-set}
\newcommand{\ball}{\mathcal{B}}	
\newcommand{\patch}[1]{\Chat^{(k_j)_{j\in #1}}_{#1}}	% patch {set of j-indices contributing to that patch}
\newcommand{\Expv}[1]{\left\langle#1\right\rangle_{\rho}}	% large langle rangle with rho
\newcommand{\subSysInds}{\tilde{S}}			% Given the subsystem in [V] this is the corresponding set of indices in [2V]

\definecolor{jens}{rgb}{.2,0.7,.9}

\definecolor{mathis}{rgb}{.9,.0,.9}

\definecolor{marek}{rgb}{.5,.5,.1}

\definecolor{gogolinblue}{rgb}{0.23,0.4,0.7}

\begin{document}
\title{Equilibration via Gaussification in fermionic lattice systems}
\author{M.\ Gluza}
\affiliation{Dahlem Center for Complex Quantum Systems, Freie Universit{\"a}t Berlin, 14195 Berlin, Germany}
\author{C.\ Krumnow}
\affiliation{Dahlem Center for Complex Quantum Systems, Freie Universit{\"a}t Berlin, 14195 Berlin, Germany}
\author{M.\ Friesdorf}
\affiliation{Dahlem Center for Complex Quantum Systems, Freie Universit{\"a}t Berlin, 14195 Berlin, Germany}
\author{C.\ Gogolin}
\affiliation{ICFO-Institut de Ciencies Fotoniques, Mediterranean Technology Park, 08860 Castelldefels (Barcelona), Spain}
\affiliation{Max-Planck-Institut f{\"u}r Quantenoptik, Hans-Kopfermann-Stra{\ss}e 1, 85748 Garching, Germany}
\author{J.\ Eisert}
\affiliation{Dahlem Center for Complex Quantum Systems, Freie Universit{\"a}t Berlin, 14195 Berlin, Germany}
\date{\today}
\begin{abstract}
In this work, we present a result on the non-equilibrium dynamics causing equilibration and Gaussification of quadratic non-interacting fermionic Hamiltonians.
Specifically, based on two basic assumptions --- clustering of correlations in the initial state and the Hamiltonian exhibiting delocalizing transport --- we prove that non-Gaussian initial states become locally indistinguishable from fermionic Gaussian states after a short and well controlled time.
This relaxation dynamics is governed by a power-law independent of the system size.
Our argument is general enough to allow for pure and mixed initial states, including thermal and ground states of interacting Hamiltonians on and large classes of lattices as well as certain spin systems.
The argument gives rise to rigorously proven instances of a convergence to a generalized Gibbs ensemble.
Our results allow to develop an intuition of equilibration that is expected to be more generally valid and relates to current experiments of cold atoms in optical lattices.
\end{abstract}

\maketitle

Despite the great complexity of quantum many-body systems out-of-equilibrium, local expectation values
in such systems show the remarkable tendency to equilibrate to stationary values that do not depend on the microscopic details of the initial state,
but rather can be described with few parameters using thermal states or generalized Gibbs ensembles \cite{PolkovnikovReview,1408.5148,1503.07538}.
Such behavior has been successfully studied in many settings theoretically and experimentally, most notably in instances of quantum simulations in 
optical lattices \cite{BlochSimulator,1408.5148,nature_bloch_eisert}.

By now, it is clear that, despite the unitary nature of quantum mechanical evolution, local expectation
values equilibrate due to a dephasing between the eigenstates
\cite{CramerEisert,Linden_etal09,1110.5759,RigolFirst,ReimannKastner12,PhysRevE.90.012121,1503.07538}.
So far it is, however, unclear why this dephasing tends to happen so rapidly.
In fact, experiments often observe equilibration after very short times which are independent of the system size \cite{nature_bloch_eisert,Koehl3}, while even the best theoretical bounds for general initial states of concrete systems diverge exponentially \cite{PhysRevE.90.012121,1408.5148}.
This discrepancy poses the challenge of precisely identifying the equilibration time, which constitutes one of the main open questions in the field \cite{1408.5148,1503.07538,PolkovnikovReview}.

What is more, only little is known about how exactly the equilibrium expectation values emerge.
Due to the exponentially many constants of motion present in quantum many-body systems, corresponding to the
overlaps with the eigenvectors of the system, there seems to be no obvious reason why equilibrium values often
only depend on few macroscopic properties such as temperature or particle number.
In short: It is unclear how precisely the memory of the initial conditions is lost during time evolution.

%\begin{figure}[t] 
%  \includegraphics[width=0.75\columnwidth]{visuallyarrestingplot.pdf}
%  \caption[visuallyarresting]{Illustration of the Gaussification process
%Expanding a time evolving local operator in the basis of creation and
%annihilation operators allows us to separate Gaussian (two body, dark
%blue) and non-Gaussian (multi-body, bright red) contributions to its
%expectation value in states with exponential clustering of
%correlations
%Delocalizing transport leads to an algebraic suppression
%of the non-Gaussian contributions in time.}
%  \label{fig:visuallyarresting}
%\end{figure}

To make progress towards a solution of these two problems, 
it is instructive to study the behavior of non-interacting particles captured by so-called quadratic or free models.
In these models the time evolution of so called Gaussian states, which are fully described by their correlation
matrix, is particularly simple to describe.
While studying the time evolution of such states provides valuable insight into the spreading of particles and equilibration, 
it is unclear if and under which conditions general non-Gaussian initial states out of equilibrium end up appearing Gaussian.

In this work, we address this question: We show under which conditions very general non-Gaussian initial states become locally indistinguishable from Gaussian states with the same second-moments.
This mechanism is much reminiscent of actual thermalization, in that an initially complex setting appears to converge to a high-entropy state that is defined by astoundingly few parameters only.
In this way, we  present a significant step forward in the theory of equilibration of quantum many-body systems that have been pushed out of equilibrium.
Furthermore, our work suggests that for quadratic models, Gaussification can be seen as a genuine mechanism of non-equilibrium dynamics, complementing and playing a significant role in equilibration.

Our results hold for a remarkably large class of initial states, including ground states of interacting 
models evolving, after a so-called \emph{quench}, in time under a quadratic fermionic Hamiltonian with finite ranged interactions.
This family of Hamiltonians notably includes the case of non-interacting ultra-cold fermions realizable in optical lattices.
By virtue of the Jordan-Wigner transformation our results also apply to certain spin models.
We formulate our results in the form of a rigorously proven theorem, which at the same time provides an intuitive explanation of the physics behind our result.
In particular, we find Gaussification to be a consequence of two natural assumptions,
namely exponential clustering of correlations in the initial state and what we call delocalizing transport.

\paragraph*{Setting.} We begin by precisely stating the physical setting that we consider.
Let $\lat$ be a $\dlat$-dimensional cubic lattice with $\vol$ lattice sites.
Each site $r\in\lat$ is associated with a fermionic orbital with fermionic creation and annihilation operators $f_r^\dagger$ and $f_r$.
We collect them in a vector $c=(f_1,f_1^\dagger,\dots, f_\vol, f_\vol^\dagger)$.
We restrict to spin-less fermions on cubic lattices purely for notational convenience, and all results can be generalized to fermions with internal degrees of freedom on Kagom\'e, honeycomb, or other geometries.
The Hamiltonian of a quadratic fermionic system is then of the form
\begin{equation}\label{QuadraticForm}
  H = \sum_{j,k=1}^{2 \vol} c_j^\dagger\, h_{j,k}\, c_k,
\end{equation}
where $h = h^\dagger$ collects the couplings.
The time evolution of annihilation operators in the Heisenberg picture under such a Hamiltonian is given by 
\begin{equation}
  \label{eq:time_evolution}
  c_{j}(t) \coloneqq \e^{\iu Ht}\, c_j\, \e^{-\iu Ht} = \sum_{k=1}^{2 \vol} W_{j,k}(t)\, c_k
\end{equation}
with the propagator $W(t) \coloneqq \e^{-2\iu t h}$.

Next we introduce the concept of Gaussian states and Gaussification.
Define the correlation matrix $\gamma$ of a state $\rho$ as the matrix of its second moments, i.e., $\gamma_{j,k} \coloneqq \tr(\rho\,c_j^\dagger\,c_k)$.
A convenient characterization of Gaussian states is the following:
They are the states that maximize the von Neumann entropy given the expectation values collected in the correlation matrix.
For every state $\rho$, we hence define its Gaussified version $\rho_G$ as the Gaussian state with the correlation matrix of $\rho$, i.e., 
$\tr(\rho_G\,c_j^\dagger\,c_k) = \tr(\rho\,c_j^\dagger\,c_k)$.

\paragraph*{Assumptions.}
Our main theorem holds for initial states (including non-Gaussian ones) with a form of decay of correlations that evolve under quadratic Hamiltonians that exhibit a form of transport that we define below.
We now make these two conditions precise, starting with the correlation decay:
\begin{dfn}[Exponential clustering of correlations]
  \label{def:clustering}
  We say that a state $\rho$ exhibits exponential clustering of correlations with length scale $\xi>0$ 
  and constant $\Cclust>0$ if, for any two operators $A,B$ with $\|A\| = \|B\| = 1$, we have
  \begin{equation}
    \begin{split}
      &|\tr(\rho\,A\,B) - \tr(A\,\rho)\, \tr(B\,\rho)| \\
      \leq{} &\Cclust\, |\supp(A)|\,|\supp(B)|\, \e^{- \dist(A,B) / \xi} .
    \end{split}
  \end{equation}
\end{dfn}
Here $\dist(A,B)$ is taken to be the natural distance on the lattice between the supports $\supp(A),\supp(B)$ of $A$ and $B$ and $\|\cdot\|$ denotes the operator norm.

Ground states of interacting gapped local Hamiltonians \cite{math-ph/0507008,Nachtergaele2013} as well as thermal states of arbitrary non-quadratic fermionic systems \cite{Kliesch2014} at sufficiently high temperature have exponential clustering of correlations as defined in Definition~\ref{def:clustering}.
Thus the initial state could be prepared within a quench scenario where the Hamiltonian is changed from a gapped interacting model to a quadratic Hamiltonian which governs the non-equilibrium dynamics.
To reemphasize, by no means is the initial state assumed to be in any specific relation to properties of the latter quadratic Hamiltonian.

For our proof of local relaxation towards a Gaussian state we further assume that the quadratic Hamiltonian exhibits transport in the following sense:
\begin{dfn}[Delocalizing transport]
  \label{def:transport}
  A quadratic Hamiltonian with propagator $W$ on a $\dlat$-dimensional cubic lattice of volume $\vol$ 
  exhibits delocalizing transport 
  with constants $\Chom > 0$, $\exphom>\dlat/4$ and recurrence time $\trec >0$ if, 
  for all $t \in (0,\trec]$, we have that 
  \begin{align}\label{eq:DefTransport}
    \forall j,k: \quad |W_{j,k} (t)| \leq \Chom \max \{t^{-\exphom},\vol^{-\exphom}\} .
  \end{align}
\end{dfn}
The intuition behind this definition is that an initially localized fermionic operator will spread over
a large area, such that its component on a single localized operator is dynamically suppressed.
Such a suppression can be shown to hold for important classes of models (see the appendix for further details). In particular delocalizing transport with $\exphom = \dlat/3$ can be proven for quadratic hopping Hamiltonians with constant on-site potential (see also Fig.~\ref{fig:cone}) and the critical Ising model.
The recurrence time takes into account that any non-trivial bound of the form \eqref{eq:DefTransport} is eventually violated due to the recurrent nature of the dynamics of finite dimensional quantum systems.
For quadratic, free hopping Hamiltonians, it can be shown that the recurrence time grows at least like $\vol^{6/7\dlat}$ with the system size, 
but it is expected to be usually exponentially large.

\begin{figure}
  \includegraphics[width=\columnwidth]{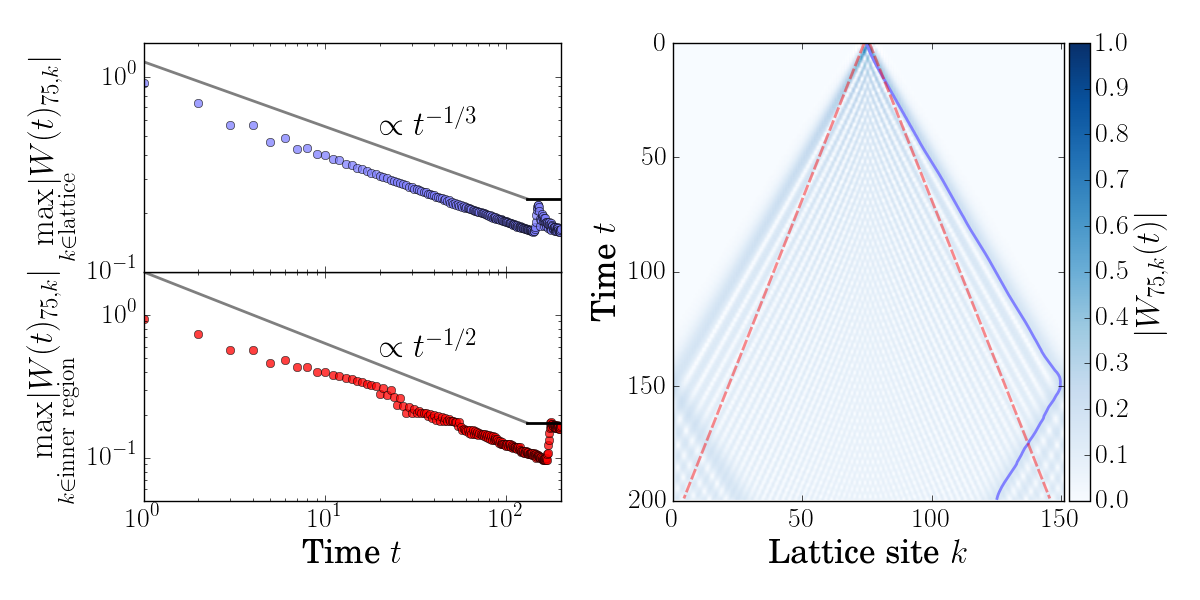}
  \caption{(color online)
    The right panel shows a numerical study of the spreading of the support of a single fermionic annihilation operator described in Eq.~\eqref{eq:time_evolution} 
    under the evolution of the quadratic hopping Hamiltonian $H = - \sum_j (f_j^\dagger f_{j+1} + f_{j+1}^\dagger f_{j})$
    on a one-dimensional chain of $150$ sites with periodic boundary conditions.
    The support expands ballistically, creating the Lieb-Robinson cone.
    The left panel shows the suppression of different elements of the propagator in time.
    The plot at top shows the evolution of the maximimum taken over the full lattice $\max_k |W_{75,k}(t)|$ where in the lower plot the maximum inside the inner region of the Lieb-Robinson cone (between the red dashed lines) is plotted. 
    The maximum taken over the full lattice is reached for $k$ in the wavefront (indicated by the blue smoothed curve in the right panel) and the suppression goes as $t^{-1/3}$, while in the bulk of the cone the suppression is proportional to $t^{-1/2}$.
    The suppression stabilizes, once the wavefronts collide.}
  \label{fig:cone}
\end{figure}

\paragraph*{Main result.}
With the above definitions our main result can be stated as follows:
\begin{thm}[Gaussification in finite time]
  \label{thm:gauss}
  Consider a family of systems on cubic lattices of increasing volume $\vol$.
  Let the initial states exhibit exponential clustering of correlations and let the Hamiltonians be quadratic finite range and have delocalizing transport with the corresponding constants $\xi,\Cclust,\Chom, \exphom$ independent of $\vol$.
  Then for any local operator $A$ on a fixed finite region and any $0 < \nu < 4\exphom-\dlat$ there is a $\vol$ independent constant $C_{\mathrm{Total}}$ such that for any $t\leq\min(\trec,V)$
  \begin{equation}
    |\tr[A(t)\,\rho] - \tr[A(t)\,\rho_G]| \leq C_{\mathrm{Total}} \, t^{-4\exphom +\dlat + \nu} .
  \end{equation}  
  Consequently, if the recurrence time $\trec$ increases unbounded as some function of $\vol$, then, given an error $\epsilon>0$ 
  there exists a relaxation time $\trelax > 0$ independent of the system size
  such that for all times $t \in [\trelax,\,\trec]$ it holds 
  that $|\tr[A(t)\,\rho] - \tr[A(t)\,\rho_G]| \leq \epsilon$.
\end{thm}

The theorem states that for all times in the interval $[\trelax,\,\trec]$ the expectation value of any local observable in the time evolved state $\rho(t)$ will agree up to an error $\epsilon$ with the expectation value in the Gaussian state $\rho_G(t)$, which has the same second moments as $\rho(t)$.
With this we find that the expectation values of all local observables can be approximated by a decomposition according to Wick's theorem and that it will be impossible to distinguish the true state $\rho$ from the fermionic Gaussian state $\rho_G$ by any local measurement on a fixed finite local region $S$.
Note that since $\trelax$ is independent of the system size, but $\trec$ increases with its volume,
for any arbitrarily small $\epsilon$ there always exists a system size such that $\trec > \trelax$ and 
the interval where the theorem applies grows as a function of the system size.

We compare our general, rigorous, analytical result with a numerical simulation in Fig.~\ref{fig:tns} that shows an experimentally detectable signature of the Gaussification of a density-density correlator.
The comparison reveals that our bound correctly reproduces the physical behavior in the sense that the true Gaussification dynamics, in the relevant situation considered, follows a power-law.
What is not correctly reproduced is the exponent of that power-law, but we understand where the discrepancy between the observed $t^{-1}$ decay and the $t^{-1/3 + \nu}$ (for arbitrarily small $\nu$) behavior of our bound with $\exphom =1/3$ originates from:
The reason is that the provable decay with $\exphom=1/3$ for the considered model roots in the slow decay of the matrix elements of the propagator at the wavefront of the Lieb-Robinson cone. 
The elements in the bulk of the Lieb-Robinson cone can numerically be found to be suppressed as $t^{-1/2}$ leading to an effective $\exphom=1/2$ for the vast majority of matrix elements (see Figure~\ref{fig:cone}). Assuming this effective $\exphom = 1/2$ in Theorem~\ref{thm:gauss} leads to a suppression as $t^{-1}$.

\begin{figure}[t]
  \includegraphics[width=\columnwidth]{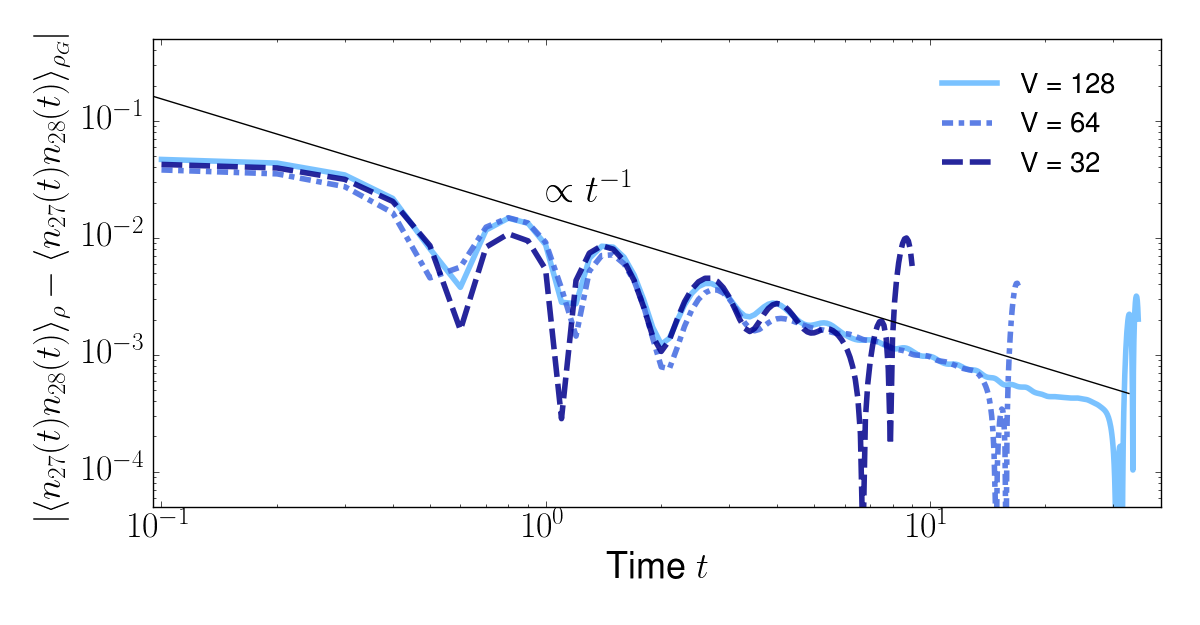}
  \caption{(color online)
    Numerical study of the evolution of the nearest-neighbor density-density correlator for system sizes $\vol=32,64,128$ under the quadratic nearest neighbor hopping Hamiltonian $H = - \sum_j^N (f_j^\dagger f_{j+1} + f_{j+1}^\dagger f_{j})$.
    The initial states are obtained by a ground state search with a matrix-product state algorithm in the interacting spinless Fermi-Hubbard model $H_{\mathrm{FH}} = H + U \sum_j^N {n_j} n_{j+1} + \sum_j \omega_j n_j$ with interaction strength $U=2$, and weak on-site disorder $w_j$ drawn independently from a Gaussian distribution with variance $1/4$. 
    To allow for comparisons of different system sizes, we have initially drawn 128 random numbers $w_j$ once and then used the first $\vol$ of them for different system sizes.
    All calculations were performed with periodic boundary conditions at half filling.
    The difference between the expectation values in the state $\rho$ and its Gaussified version $\rho_G$ of the density-density correlator between sites $27$ and $28$ as a function of time is suppressed approximately like $t^{-1}$, as indicated by the black line.
    At late times, due to the finite size of the system, recurrences occur, leading to an increase of the difference.
    Increasing the system size only shifts the recurrence time $\trec$ but leaves the decay behavior unchanged.
    The visible oscillations depend on details of the model and initial state.
    The time evolution was performed by explicitly calculating the evolution in the Heisenberg picture using Eq.~\eqref{eq:time_evolution}.
  }
  \label{fig:tns}
\end{figure}

The argument can be understood as a fermionic quantum central limit 
theorem \cite{Hudson} emerging from a dynamical evolution, in that the second moments are preserved and the higher 
cumulants can be proven to converge to zero in time.
The key steps in the proof, to be presented in the following, can be understood from intuitive physical considerations.
They are based on three main ingredients: finite speed of propagation in lattice systems, homogeneous suppression 
of matrix elements of the propagator due to delocalizing transport, and exponential clustering of correlations in the initial state.
The proof shares some intuition put forth on the equilibration of harmonic classical chains \cite{Spohn}.
For the full proof with all details of the involved combinatorics we refer the reader to the appendix.

\begin{proof}
  We expand the local operator $A$ supported in a fixed finite region $S$ in the basis of fermionic operators.
  To that end let $\subSysInds \coloneqq \{s_r\}$ for $r \in [2 |S|] \coloneqq \{1,\dots,2 |S|\}$ be the set of indices of all elements of the vector $c$ with support in $S$, then
  \begin{align}
    \label{eq:decomposeAMainText}
    A(t) = \quad \smashoperator{\sum_{b_1, \dots, b_{2 |S|}=0}^1}\quad  a_{b_1,\cdots,b_{2 |S|}} \,  c_{s_1}(t)^{b_1} \dots c_{s_{2 |S|}}{(t)}^{b_{2 |S|}} \; .
  \end{align} 
  Without loss of generality we can assume normalization $\|A\|=1$, such that all of the $2^{2 |S|}$ coefficients satisfy $|a_{b_1,\dots,b_{2m}}| \leq 1$.
Thus
  \begin{align}
    \label{eq:boundonlocalnongaussianity}
    &|\tr[A\,\rho(t)] - \tr[A\,\rho_G(t)]|\\\nonumber
    \leq{} & 2^{2|S|} \max_{J \subset \subSysInds} \biggl| \quad\smashoperator{\sum_{\substack{\phantom{X}\\(k_{j})_{j\in J}\in [2\vol]^{\times|J|}}}}\quad\tr\biggl[ \prod_{j \in J} W_{j,k_j}(t)\, c_{k_j} \, (\rho - \rho_G)\biggr]\biggr| \; .
  \end{align}
  Here and in the following all products are meant to be performed in increasing order.
  
  We assumed that the Hamiltonian has finite range interactions, i.e., there exists a fixed length scale $l_0$, such that whenever $\dist(j,k) > l_0$ it holds that $h_{j,k} = 0$, where we have used the shorthand $\dist(j,k) \coloneqq \dist(c_j,c_k)$.
  Such models satisfy Lieb-Robinson bounds \cite{liebrobinson}, which in our setting can be stated as follows:
  \begin{lm}[Lieb-Robinson bound for quadratic systems \cite{Hastings2004a}]\label{lieb-robinson-bound-lemma}
    For any quadratic fermionic Hamiltonians $H$ with finite range interactions there exist constants $C_{\mathrm{LR}},\mu,v>0$ independent of the system size such that its propagator $W$ fulfills the bound
    \begin{align}
      |W_{j,k}(t)| \leq C_{\mathrm{LR}}\, \e^{\,\mu\,( v |t| -  \dist(j,k))} \; .
    \end{align}
  \end{lm}
  The Lieb-Robinson bound tells us that $c_j(t)$ and $c_k(t)$ essentially still have disjoint support as long as $t$ is small enough such that $v |t| \ll\dist(j,k)$.
  We can hence restrict the sum in Eq.~\eqref{eq:boundonlocalnongaussianity} to those $k_{j}$ whose $\min_{s \in \subSysInds} \dist(k_j,s)$ is smaller than $(v+2\,v_\epsilon) |t|$ for some fixed $v_\epsilon >0$.
  The total contribution of the neglected terms can be bounded explicitly and, 
  importantly, is independent of $\vol$ and exponentially suppressed in $|v_\epsilon\,t|$.

  For each of the remaining summands in Eq.~\eqref{eq:boundonlocalnongaussianity} it is now important to keep track of the distribution of the indices $k_j$ inside the cone.
  For this purpose we define the $\Delta$-partition $P_\Delta$ of a subindex set $J \subset \subSysInds$ and sequence of indices $(k_j)_{j_\in J}$ as the unique decomposition of $J$ into subsets (patches) $p$ in the following way:
  The patches are constructed such that for any two subindices within any given patch $p$ there is a connecting chain of elements from that patch in the sense that the distance between two consecutive $c_{k_j}$ with $j \in p$ along that chain is not greater than $\Delta$ and the distance between any two $c_{k_j}, c_{k_j'}$ with $j,j'$ from different patches is larger than $\Delta$.
  For each patch $p$ in the $\Delta$-partition of a given summand in Eq.~\eqref{eq:boundonlocalnongaussianity} we define a corresponding operator 
  \begin{align}
    \patch{p} \coloneqq \prod_{j \in p} W_{j,k_j}(t)\, c_{k_j}  \; .
  \end{align}
  We can then reorder the factors in Eq.~\eqref{eq:boundonlocalnongaussianity} to write the product as a product over these operators.
  The exponential clustering of correlations (Definition~\ref{def:clustering}) in the initial state allows us to factor the patches if we scale $\Delta$ suitably with $|t|$.
  Concretely, for $\sigma \in \{\rho,\rho_G\}$ the expectation values appearing in Eq.~\eqref{eq:boundonlocalnongaussianity},
  which we denote by $\langle \cdot\rangle_{\sigma}$, can be approximated as follows
  \begin{equation} \label{eq:factoredintopatches}
    \biggl\langle \prod_{p \in P_\Delta} \patch{p} \biggr\rangle_{\sigma} \approx \prod_{p \in P_\Delta} \langle \patch{p} \rangle_{\sigma} \; .
  \end{equation}
  The error thereby introduced is exponentially suppressed with the ratio of patch distance and correlation length $\Delta/\xi$.
  
  It remains to bound the contribution from the factorized patches that are completely inside the Lieb-Robinson cone.
  Note that the right hand side of Eq.~\eqref{eq:factoredintopatches} can be non-zero only if all the patches are of even size, as $\rho$ and $\rho_G$ have an even particle number parity.
  Moreover, as the second moments of $\rho$ and $\rho_G$ are equal, the difference of the right hand side for $\sigma = \rho$ and $\sigma = \rho_G$ vanishes whenever all patches have size $2$.
  Hence, only partitions that contain at least one patch of size at least $4$ can contribute.
  The delocalizing transport of the Hamiltonian implies that the contribution from such larger patches however is dynamically suppressed.
  Whenever $|p| \geq 4$ it holds that
  \begin{equation}
    | \langle \patch{p} \rangle_{\sigma} | \leq \Chom^4\, t^{-4\exphom}  
  \end{equation}
  as long as $V$ is large enough.
  The influence of possible patches of size $2$ in the same decomposition makes it necessary to bound the overall contribution with an involved recursive and combinatorial argument.
  However, effectively the dynamical suppression stated in the last inequality allows us to derive a bound that increases with the patch size $\Delta$ but is algebraically suppressed in time $t$.
  The increase with $\Delta$ is a consequence of the fact that allowing for longer distances between the elements of a patch increases the number of possible patches of a given size.
  Finally, by choosing $\Delta = \max(1,t^{\nu/4\dlat})$ for some $0 < \nu < 4\exphom-\dlat$, one can obtain an at least algebraic suppression with $t$ of all terms and thereby of the difference $|\tr(A(t)\rho) - \tr(A(t)\rho_G)|$.
\end{proof}
\paragraph*{Physical implications and applications.}
\label{physicalimplications}
The Gaussification result presented above also has profound implications for the study of equilibration of quantum many-body systems.
Whenever the second moments equilibrate, which is often observed
\cite{nature_bloch_eisert,CramerEisertScholl08,CalabreseCardy06,1205.2211,0906.1663,1302.6944,0906.1663,PeriodicDriven_2}, our results imply that the full reduced density matrix becomes stationary.
The numerical study presented in Fig.~\ref{fig:tns} shows that the power-law appearing in Theorem~\ref{thm:gauss} is not an artifact of our proof strategy 
but reflects the underlying physics, that can moreover be observed in experiments.
The quadratic models considered here constitute a ``theoretical laboratory'', in which
the mechanisms of Gaussification and equilibration can be very precisely and quantitatively characterized, and all specifics of the processes laid out.
This does not mean that the physics we address is very specific to precisely these quadratic Hamiltonians:
Quite to the contrary, we expect the fundamental mechanisms underlying the result --- local relaxation due to transport and initial clustering of correlations --- to be generic, fundamental and ultimately the reason for relaxation in a wide classes of interacting models \cite{Muramatsu}.
The intuition, reminiscent of a quantum central limit theorem, that incommensurate influences of further and further separated regions lead to mixing and relaxation is then expected to still be valid.
It is also important to stress that our main theorem equally applies to mixed initial states, such as thermal states, which are relevant in present day experiments with 
ultra-cold fermions \cite{Schneider_fermionic_transport, GreinerFermions,Koehl2,KuhrFermions}.

Returning to the specifics of quadratic Hamiltonians, the result derived here can be interpreted in yet another way: 
It is reminiscent of the initial state converging towards a \emph{generalized Gibbs ensemble} (GGE) \cite{RigolFirst,1205.2211,PeriodicDriven_1} in the sense that the initial state becomes close to a Gaussian state, which is the maximum entropy state given the second moments.
Different from a real GGE, the observables $\{I_\alpha\}$ held fixed while maximizing entropy can be time-dependent, i.e., $\rho_G(t)=\exp(\sum_\alpha \lambda_\alpha I_\alpha(t))/\mathcal{Z}$.
Here $\{\lambda_\alpha\}$ are appropriately chosen Lagrange multipliers, $\mathcal{Z}$ the partition function, and $\{I_\alpha(t)\}$ the number operators of the eigenmodes of $\gamma$.
However, in the case of equilibrating second moments all relevant $I_\alpha$ become
time-independent such that our theorem constitutes a proof of a convergence to a proper GGE in these cases.
The same holds true for integrable spin models that can be mapped to the type of fermionic models considered here,
complementing insights on bosonic systems \cite{CramerEisert,NJP}.

\paragraph*{Conclusion and outlook.}
In this work we have established an understanding of how systems quenched to non-interacting
fermionic Hamiltonians locally converge to Gaussian states.
Out of equilibrium 
dynamics is identified as having the tendency to bring systems locally
in maximum entropy states given the second moments.
This holds even if the initial state
was far from being a Gaussian state, e.g., a ground state of a strongly interacting model.
This is achieved based on just two natural assumptions: A form of delocalizing transport in the model and exponential clustering of correlations in the initial state.
Otherwise the initial state can be completely general.
It is the hope that the present work will serve as a stepping stone to gain further insights into the relaxation dynamics of more complex quantum many-body systems and the consequences of the suppression of transport in, for example, localizing systems.

\paragraph*{Acknowledgements.} 
We acknowledge fruitful discussions with M.\ Cramer, M.\ M.\ Wolf and P. \'Cwikli\'nski.
We would like to thank the EU (RAQUEL, SIQS, AQuS, IP, QUIC), the ERC (TAQ, OSYRIS, QITBOX), the BMBF (Q.com),
the DFG (EI 519/7-1, CRC 183), the Studienstiftung des Deutschen Volkes, MPQ-ICFO, the Spanish Ministry Project FOQUS 
(FIS2013-46768-P), MINECO (Severo Ochoa Grant No.~SEV-2015-0522), Fundaci\'{o} Privada Cellex, the Generalitat 
de Catalunya (SGR 874 and 875), the Spanish MINECO (Severo Ochoa grant SEV-2015-0522), ICFOnest+ (FP7-PEOPLE-2013-COFUND), the EU's Marie Skłodowska-Curie Individual Fellowships programme under GA: 700140, and the COST Action MP1209 for support.

%\bibliography{gaussification}


\begin{widetext}
\newpage

\section*{Appendix}

\section{Majorana operators and time evolution in quadratic fermionic systems}
\label{FreeFermions}
In this appendix, we formulate the Majorana operator description that allows to conveniently derive the operator governing transport in the system. 
We introduce the Majorana operators as
\begin{align}
  m_{2j-1} &\coloneqq (f_j^\dagger + f_j)/\sqrt{2},\\
  m_{2j} &\coloneqq \iu\, (f_j^\dagger - f_j)/\sqrt{2},    
\end{align}
which are collected in a vector $m=(m_1,\dots, m_{2\vol})$
\cite{cond-mat/0506438}.
The vector $c$ of creation and annihilation operators used in the main text and $m$ are related by the unitary transformation 
\begin{equation}
  \Omega \coloneqq \frac{1}{\sqrt{2}}\bigoplus_{j=1}^\vol 
  \left(
    \begin{array}{cc}
      1& 1\\
      -\iu & \iu
    \end{array}
  \right),
\end{equation}	
as $m= \Omega\,c$.
The Majorana operators are Hermitian and satisfy the anti-commutation relations $\{m_j,m_k\}= \delta_{j,k}$ for $j,k \in [2\vol] \coloneqq \{1,\dots, 2\vol\}$.
The algebra generated by those operators constitutes a Clifford algebra. Linear transformations of the form
\begin{equation}
  m_j' = \sum_{j,k=1}^{2\vol} O_{j,k}\, m_k, \quad O\in SO(2\vol)
\end{equation}
transform a vector of legitimate Majorana operators to a new such vector.

The most general form of a Hamiltonian considered in this work can be written in terms of the Majorana operators as follows
\begin{equation}
  H = \iu \sum_{j,k=1}^{2\vol} m_j\, K_{j,k}\, m_k \; ,
\end{equation}	
where $K=-K^T$ is real and anti-symmetric.
It is straightforward to relate such Hamiltonians to the ones expressed in the form of the main text.
The kernel $K$ can be obtained form $h$ via $K=-\iu\,\Omega\,  h\, \Omega^\dagger$.

Time evolution can be captured conveniently in the Majorana operator formulation.
Using the Baker-Campbell-Hausdorff formula, that $K$ is anti-symmetric, and the algebraic structure of the Majorana fermions, one arrives at the following expression for their time evolution in the Heisenberg picture
\begin{equation} \label{eq:definitonL}
  m_j(t) \coloneqq \e^{\iu Ht} m_j \e^{-\iu Ht} = \sum_{k=1}^{2\vol} (\e^{2t K})_{j,k}\, m_k = \sum_{k=1}^{2\vol} \L_{j,k}(t)\, m_k,
\end{equation}
where $\L(t) \coloneqq \e^{2 t K}$. Now notice that as the propagator defined in the main text is related to $L(t)$ via
\begin{equation} \label{eq:defW}
  W(t) = \Omega^\dagger\,\L(t)\,\Omega
\end{equation}
and hence
\begin{equation}
  c_j(t) = \sum_{k=1}^{2\vol} W_{j,k}(t)\, c_k
\end{equation}
as claimed in the main text.

Further, we introduce some general notation which we will use in the following.
For any given operator $A$ that is supported on a region $S$ we had defined the set $\subSysInds = \{s_1, \cdots, s_{2 |S|}\}$ with $s_1<s_2<\ldots < s_{2|S|}$ the set of the indices of fermionic basis operators in $S$.
We can expand $A$ as
\begin{align}
  \label{eq:app_fullop}
  A = \sum_{b_1, \dots, b_{2 |S|}=0}^1 a(\{s_r:b_r=1\}) \,  c_{s_1}^{b_1} \dots c_{s_{2 |S|}}^{b_{2 |S|}} \; ,
\end{align}
with $a(\{s_r:b_r=1\}) = a_{b_1,\ldots,b_{2|S|}}$.
The sum in Eq.~\eqref{eq:app_fullop} goes over all possible configurations of fermionic basis operators on the region $S$.
We can hence group the summands according to the subset $J$ of indices from $\subSysInds$ for which a given term actually contains a fermionic basis operator and write it as a sum
\begin{equation}
  \label{eq:decomposeA}
  A = \sum_{J\subset \subSysInds} a(J)\,A_J \; ,
\end{equation}
with 
\begin{equation}
  A_J \coloneqq \prod\limits_{j\in J} c_j \; .
\end{equation}
Here, and whenever such expressions appear in the following, we take the product over $j\in J$ in the ordered dictated by the ordering of the lattice sites.
The time evolution of $A$ is then given by $A(t) = \sum_{J\subset \subSysInds} a(J)\,A_J(t)$ with
\begin{align}
  \label{eq:app_LRinitialop}
  A_J(t) = \sum_{\substack{(k_j)_{j\in J} \in [2V]^{\times |J|}}} \quad \biggl( \prod_{j \in J} W_{j,k_j}(t)\, c_{k_j} \biggr) \; .
\end{align}

\section{Transport}
In this appendix we show that two prototypical example systems exhibit delocalizing transport as defined in Definition 2 in the main text: the fermionisation of the Ising model with appropriate initial states and the fermionic nearest neighbor hopping model. 
Our proofs closely follow along the lines of the investigation of the transport properties of the propagator presented in Ref.\ \cite{NJP}.

\subsection{Spreading in the Ising model}
\label{Ising}
We start by considering the 1D Ising model and show that it exhibits delocalizing transport at criticality. 
Its Hamiltonian for $\vol$ sites is
\begin{align}
  H_{\mathrm{IS}} = - \sum_{j=1}^{\vol} X_j X_{j+1} - g \sum_{j=1}^\vol Z_j \; , 
\end{align}
where $X_j,Z_j$ are the Pauli matrices supported on site $j$ and $g$ is a real parameter. We adopt periodic boundary conditions.
Invoking the Jordan-Wigner transformation \cite{Lieb_JW}, this spin system can be mapped to fermions,
using the substitutions
\begin{align}
  Z_j &\mapsto f_j f_j^\dagger - f_j^\dagger f_j = 1 - 2 n_j \; ,\\
  S_j = \frac{1}{2}(X_j-\iu Y_j) &\mapsto \prod_{l<j} (1 - 2 n_l) f_j \; ,\\
  S^\dagger_j = \frac{1}{2}(X_j+\iu Y_j) &\mapsto \prod_{l<j} (1 - 2 n_l) f_j^\dagger \; ,
\end{align}
where $S_j$ is the spin annihilation operator associated with site $j$
and $n_j = f_j^\dagger f_j$ the usual fermionic number operator.
After this transformation, the Ising Hamiltonian takes the form
\begin{align}
 H_{\mathrm{IS}} &= -\sum\limits_{j=1}^{\vol-1}(f_j^\dag+f_j)(1-2n_j)(f^\dag_{j+1}+f_{j+1}) - \prod\limits_{j=1}^{\vol-1}(1-2n_j)(f_\vol^\dag+f_\vol)(f_1^\dag+f_1) - g\sum\limits_{j=1}^V (1-2n_j)\\
  &= -\sum\limits_{j=1}^{\vol-1}(f_j^\dag-f_j)(f^\dag_{j+1}+f_{j+1}) + \prod\limits_{j=1}^{\vol}(1-2n_j)(f_\vol^\dag-f_\vol)(f_1^\dag+f_1)- g\sum\limits_{j=1}^V (1-2n_j).\nonumber
\end{align}
Using the Majorana operators introduced in Appendix~\ref{FreeFermions}, we can rewrite the Hamiltonian as
\begin{equation}
 H_{\mathrm{IS}} = \sum\limits_{j=1}^{\vol-1}\iu (m_{2j} m_{2j+1}-m_{2j+1}m_{2j}) - \iu\prod\limits_{j=1}^{\vol}(1-2n_j) ( m_{2\vol}  m_1 -m_1 m_{2\vol})- \iu g\sum\limits_{j=1}^\vol (m_{2j}m_{2j-1}-m_{2j-1}m_{2j}).
\end{equation}
As Theorem 1 would not be applicable otherwise, we restrict the discussion here to initial spin-states $\rho$ which are mapped by the Jordan Wigner transformation to proper fermionic states respecting the parity super-selection rule.
In that case, the parity operator $\prod_{j=1}^\vol(1-2n_j)$ will take a fixed value $\sigma=\pm 1$, depending on the parity of $\rho$.
That is, depending on the parity sector of the state labeled by $\sigma$, $H_{\mathrm{IS}}$ is of the form
\begin{align}
  H^\sigma &=  \iu  \sum_{j,k=1}^{2\vol} m_j K_{j,k}^\sigma m_k \label{IH}
\end{align}
with
\begin{align} \label{eq:CritIsingKernel}
     K^\sigma &= 
    \begin{pmatrix}
      0 & g &&&&& \sigma \\ 
      -g & 0 & 1\\
      & -1 & 0 & g\\
      && -g & 0 & 1\\
      && & -1 & 0\\
      -\sigma&&&&&\ddots \\
    \end{pmatrix}.
  \end{align}
  In the following we consider the special case $g=1$, corresponding
  to the critical Ising model.

\begin{lm}[Delocalizing transport in the critical Ising model] \label{lm:delocalizing_transport_in_the_critical_ising_model}
    For the one-dimensional fermionic model given in Eqs.~\eqref{IH} and \eqref{eq:CritIsingKernel} corresponding to the
    critical Ising model, there is a constant $\Chom > 0$ such that for all $t\in (0,\trec]$, with $\trec = \vol^{6/7}$ the recurrence time, it holds that
    \begin{equation}
      |W_{j,k}^\sigma (t)| \leq \Chom t^{-1/3} \, \forall j,k.
    \end{equation}
\end{lm}
\begin{proof}
We diagonalize with $K^\sigma$ using a modified discrete Fourier transform
\begin{equation}
 U_{k,x}^\sigma = \frac{1}{\sqrt{2\vol}}\e^{\iu\pi k x/ 2\vol}\e^{\iu\pi(1+\sigma)(x+k)/4\vol}
\end{equation}
in order to obtain its spectrum 
\begin{equation}
 \lambda_k^\sigma = 2\sin(\pi [k + (1+\sigma)/4]/\vol).
\end{equation}
We then find for the propagator (see Eq.~\eqref{eq:definitonL})
\begin{equation}
 L_{j,l}^\sigma(t) = (e^{2tK^\sigma})_{j,l} = \frac{1}{2\vol}\sum\limits_{k=1}^{2\vol}\e^{\iu \pi [k+ (1+\sigma)/4](l-j)/\vol}f(\pi [k + (1+\sigma)/4]/\vol), \label{eq:spin_tran_X_in_terms_of_f}
\end{equation}
with $f(\phi) = \e^{2\iu t \sin(\phi)}$. The Fourier transform of $f(\phi)$ is given by $f(\phi) = \sum_{n=-\infty}^{\infty} f_n \e^{-\iu n \phi}$ with the modes $f_n = J_n(2t)$ where $J_n$ denotes the Bessel function of first kind.
By partial integration we can upper bound the absolute value of the Fourier modes of $f$ by
\begin{align}
 |f_n| = \frac{1}{2\pi n^2}\left|\int\limits_0^{2\pi}\e^{\iu n\phi}\frac{d^2}{d\phi^2}\e^{2\iu t\sin(\phi)}d\phi\right| \leq 4\frac{|t|+|t|^2}{n^2}\label{eq:spin_trans_bound_on_fourier_modes}.
\end{align}
Inserting the Fourier decomposition of $f$ into Eq.~\eqref{eq:spin_tran_X_in_terms_of_f} yields
\begin{align}
 L_{j,l}^\sigma(t) &= \frac{1}{2\vol}\sum\limits_{n=-\infty}^{\infty}\sum\limits_{k=1}^{2\vol}\e^{\iu \pi [k+ (1+\sigma)/4](l-j-n)/\vol} f_n = \sum\limits_{p=-\infty}^{\infty}(-\sigma)^p f_{l-j+2p\vol}.
\end{align}
By the periodicity of the model, i.e.~$L_{j,l}(t) = -\sigma L_{j-2\vol,l}(t) = -\sigma L_{j,l-2\vol}(t)$, we are save to assume w.l.o.g. $|j-l|\leq \vol$. Using the upper bound in Eq.~\eqref{eq:spin_trans_bound_on_fourier_modes} and upper bounding the resulting converging series for $|j-l|\leq \vol$ yields
\begin{align}
 |L_{j,l}^\sigma(t) - J_{l-j}(2t)| \leq \frac{8 (|t|+|t|^2)}{4\vol^2}\sum\limits_{p=1}^\infty  \frac{1}{[(l-j)/2\vol+ p]^2} \leq \pi^2\frac{|t|+|t|^2}{\vol^2}.
\end{align}
With the general upper bound on the Bessel function of first kind $J_n(x)\leq x^{-1/3}$ we conclude that for $L>1$
\begin{align}
 |L_{j,k}^\sigma(t)| \leq 11\, t^{-1/3}\quad \forall j,k
\end{align}
for $t\in(0,\trec]$ with $\trec = \vol^{6/7}$.
Due to the block structure of $\Omega$ in Eq.\eqref{eq:defW} it then follows directly that
\begin{align}
 |W_{j,k}^\sigma(t)| \leq 22\, t^{-1/3}\quad \forall j,k.
\end{align}
\end{proof}

\subsection{Transport in fermionic nearest neighbor hopping models in square lattices}
\label{Hopping}
We now turn to fermionic hopping models, i.e., systems whose Hamiltonian is a linear combination of terms of the form $f_j^\dagger f_k$.
Instead of the general quadratic form in the form of the main text, the Hamiltonian can then be written as
\begin{equation} \label{eq:hoppinghamiltonain}
  H = \sum_{j,k=1}^\vol f_j^\dagger M_{j,k} f_k   ,
\end{equation}
with $M$ a real and symmetric matrix, i.e., $M=M^T$.
The time evolution of fermionic annihilation operators in the Heisenberg picture is then given by
\begin{equation}
  f_j(t) =  \e^{\iu Ht}\, f_j\, \e^{-\iu Ht} = \sum_{k=1}^\vol N_{j,k}(t)\,f_k,
\end{equation}
where $N(t) \coloneqq \e^{-\iu M t}$.
To connect this to the notation used in the main text, note that with $P$ the permutation matrix that acts as
\begin{equation}
  (f_1,f_1^\dagger,\dots, f_\vol, f_\vol^\dagger)=  P (f_1,\dots, f_\vol,
  f_1^\dagger,\dots, f_\vol^\dagger).
\end{equation}
The Hamiltonian from Eq.~\eqref{eq:hoppinghamiltonain} can be written in the form of the main text by adding an appropriate constant and setting 
\begin{equation}
  h = \frac{1}{2} P \left( M\oplus -M\right) P^\dagger .
\end{equation}
$N(t)$ is then related to $W(t)$, via
\begin{equation}\label{Form}
  W(t) = P (N(t) \oplus N(t)^\dagger) P^\dagger .
\end{equation}

We now consider the particularly important case of nearest neighbor hopping on a square lattice of spacial dimension $\dlat$ with $\vol$ sites, periodic boundary conditions and hopping strength one.
For convenience, we restrict the discussion to $\vol^{1/\dlat}$ even.
Writing the Hamiltonian as in \eqref{eq:hoppinghamiltonain}, the coupling matrix $M$ of this model can be decomposed into a sum over the $\dlat$ different spatial directions as follows
\begin{equation}
  M = \sum\limits_{k=0}^{\dlat-1} \id^{\otimes k} \otimes M^{(1)} \otimes \id^{\otimes (\dlat -k-1)} ,
\end{equation}
with
\begin{equation}
  M^{(1)} \coloneqq -\begin{pmatrix}
    0 & 1 &&&&& 1\\ 
    1 & 0 & 1\\
    &1 & 0 & 1\\
    && 1& 0 &\\
    1&&&&&\ddots \\
  \end{pmatrix}\in\R^{V^{1/\dlat}\times V^{1/\dlat}}.\label{eq:FreeHoppingCoupling1D}
\end{equation}
For such models we can bound the spreading as follows:
\begin{lm}[Delocalizing transport in fermionic hopping models] \label{lm:delocalizing_transport_in_fermionic_hopping_models}
  For the fermionic nearest neighbor hopping model on a square lattice of spacial dimension $\dlat$ and $\vol$ sites with $V^{1/\dlat}$ even and periodic boundary conditions there is a constant $\Chom>0$, independent of the volume $\vol$, such that for all $t \in (0, \trec]$, with $\trec = \vol^{6/7\dlat}$ the recurrence time, it holds that
  \begin{equation}
    \forall j,k  \qquad |W_{j,k} (t)| \leq \Chom\,t^{-\dlat/3} .
  \end{equation}
\end{lm}
\begin{proof}
  From the structure of the Hamiltonian it follows that $N(t) = (\e^{-\iu M^{(1)} t})^{\otimes\dlat}$. 
  For $V^{1/\dlat}$ even we obtain that 
  \begin{align}
   -\iu K^{\sigma}(t) = Q^\dag M^{(1)}(t) Q
  \end{align}
  with $Q = \textrm{diag}(1,i,-1,-i,1,\ldots)$ where $\sigma = -1$ if $\vol^{1/\dlat}/2$ is even and $\sigma = 1$ otherwise. 
  From the proof of Lemma~\ref{lm:delocalizing_transport_in_the_critical_ising_model} we then obtain for all $0<t<\trec = 2\vol^{6/7\dlat}$
  \begin{align}
    \forall j,k  \qquad |N_{j,k}(t)| \leq 2^{\dlat/3} 11^{\dlat} t^{-\dlat/3}.
  \end{align}
  This bound is inherited by $W$ as $N$ and $W$ are related by a permutation of rows and columns.
\end{proof}

\section{Fermionic Gaussian states and clustering of correlations}
\label{GaussianClusteringOfCorrelations}
In this appendix, we provide some background on fermionic Gaussian states. We first demonstrate that they are the maximum entropy states
given their correlation matrix. Moreover, we show that whenever a state has exponential clustering of correlations, then its Gaussified version also shows an exponential correlation decay.

\subsection{Gaussian states as maximum entropy states}\label{MaxE}

In this subsection, we show that fermionic Gaussian states are the maximum entropy states given the second moments of fermionic operators. This in particular highlights that the state to which we show apparent local convergence has the characteristic feature of generalized Gibbs ensembles (GGE) that it is the maximum entropy state given the expectation value of a set of observables, being the second moments here.

\begin{lm}[Gaussian states as maximum entropy states]
For a given correlation matrix $\gamma \in \C^{2V\times 2V}$,
\begin{equation}
	\rho_G = \argmax_\rho \left\{
	S(\rho) : \gamma(\rho)=\gamma
	\right\},
\end{equation}
where the maximization is performed over all quantum states $\rho$ and $\gamma(\rho)$ denotes the correlation matrix of a state $\rho$.
\end{lm}
\begin{proof}
This statement follows immediately from the positivity of the quantum relative entropy. For an arbitrary state $\rho$ and
the Gaussian state $\rho_G$ with the same correlation matrix, we have
\begin{equation}
	0\leq S(\rho\|\rho_G)  = -S(\rho) - \tr(\rho \log (\rho_G)).
\end{equation}
Since $\rho_G$ is a Gaussian state, it can be written as $\rho_G= e^{H}$ for a suitable Hermitian operator $H$ that is quadratic in the 
fermionic basis operators, which means that 
\begin{equation}	
\tr(\rho \log (\rho_G)) = \tr(\rho_G \log (\rho_G)) = -S(\rho_G ),
\end{equation}
from which the assertion follows.
\end{proof}

\subsection{Clustering of correlations of Gaussified states}
In this subsection we show that, given a state $\rho$ that exhibits exponential clustering of correlation as defined in Definition 1, its Gaussified version $\rho_G$ inherits the exponential clustering of correlations with a changed scaling of the pre-factor with the support of the considered operators.
We prove this statement by means of Wick's theorem, which connects higher to second moments for general Gaussian states.
Precisely, Wick's theorem can be stated as follows.

\begin{lm}[Wick's theorem \cite{Kraus_WicksTheorem}]\label{lm:wick}
A Gaussian state $\rho_G$ fulfills
\begin{equation}
 \tr[\prod\limits_{k=1}^n c_{i_k} \rho_G] = \Pf(\gamma^c[i_1,\ldots,i_n])\ ,
\end{equation}
where
\begin{equation}
 \label{eq:defgammac}
 \gamma^c[i_1,\ldots,i_n]_{a,b}=\begin{cases}
                                 \tr(c_{i_a}c_{i_b}\rho_G)\quad&\text{for } a<b, \\
                                 -\tr(c_{i_b}c_{i_a}\rho_G)\quad&\text{for } b<a, \\
                                 0&\text{else}.
                                \end{cases}
\end{equation}
\end{lm}

Given a state with exponential clustering of correlations also its Gaussified version will show clustering of correlations in following sense:
\label{app_gaussian}
\begin{lm}[Weak clustering of correlations for Gaussified states]
  \label{lm:clusteringofcorrelationsgauss}
  Let $\rho$ be a state that exhibits exponential clustering of correlations according to Definition 1 with constants $\Cclust,\xi>0$, then for all operators $A,B$ with $\|A\| = \|B\| = 1$ its Gaussified version $\rho_G$ satisfies 
  \begin{equation}
    |\tr(\rho_G\,A\,B) - \tr(A\,\rho_G)\, \tr(B\,\rho_G)| \leq \Cclust\, 4^{|\supp(A)|+|\supp(B)|}(|\supp(A)|+|\supp(B)|)^{|\supp(A)|+|\supp(B)|}\e^{-\dist(A,B)/\xi} .      
  \end{equation}
\end{lm}
\begin{proof}
  Note that we can assume without loss generality $\Cclust \e^{-\dist(A,B)/\xi} \leq 1$ as otherwise the trivial bound $|\tr(\rho_G\,A\,B) - \tr(A\,\rho_G)\, \tr(B\,\rho_G)|\leq2$ concludes the proof.

  We decompose a general operator supported on $\supp(A)$ and $\supp(B)$ as in Eq.~\eqref{eq:decomposeA} into the fermionic operator-basis
  \begin{equation}
   A = \sum\limits_{K\subset \supp(A)} a(K) \prod\limits_{k\in K}c_k
  \end{equation}
  and 
    \begin{equation}
   B = \sum\limits_{J\subset \supp(B)} b(J) \prod\limits_{j\in J}c_j
  \end{equation}
  correspondingly. 
  From $\|A\|=1=\|B\|$ it follows that $|a(K)|\leq 1$ and $|b(J)|\leq 1$ for all $J$ and $K$. Using the triangle inequality, we can therefore write
  \begin{equation}
   |\tr(\rho_G\,A\,B) - \tr(A\,\rho_G)\, \tr(B\,\rho_G)| \leq 2^{|\supp(A)\cup\supp(B)|}\max\limits_{\substack{K\subset\supp(A),\\J\subset\supp(B)}} 
   \left|\tr(\rho_G\,\prod\limits_{k\in K}c_k\prod\limits_{j\in J}c_j) - \tr(\prod\limits_{k\in K}c_k\,\rho_G)\, \tr(\prod\limits_{j\in J}c_j\,\rho_G)\right|.
  \end{equation}
  Let $J^\prime$ and $K^\prime$ be the sets for which the maximum is attained.
  Wick's theorem then allows us to write the expectation values in terms of second moments
  \begin{equation}
  \label{eq:gausscorrproofpfaffians}
   |\tr(\rho_G\,A\,B) - \tr(A\,\rho_G)\, \tr(B\,\rho_G)| \leq 2^{|\supp(A)\cup\supp(B)|} |\Pf\gamma^c[(k)_{k\in K^\prime},(j)_{j\in J^\prime}]-\Pf\gamma^c[(k)_{k\in K^\prime}]\,\Pf\gamma^c[(j)_{j\in J^\prime}]|.
  \end{equation}
  From the definition of $\gamma^c$ in Eq.~\eqref{eq:defgammac} it follows that $\gamma^c[(k)_{k\in K^\prime},(j)_{j\in J^\prime}]$ decomposes into blocks as follows
  \begin{equation}
   \gamma^c[(k)_{k\in K^\prime},(j)_{j\in J^\prime}] = \gamma^c[(k)_{k\in K^\prime}]\oplus\gamma^c[(j)_{j\in J^\prime}] + 
   \begin{pmatrix}
    0 &E\\ -E^T &0
   \end{pmatrix}\; .
  \end{equation}
  As $E$ contains only second moments of which one operator is supported on $\supp(A)$ and the other on $\supp(B)$ we obtain from the exponential clustering of correlations of $\rho$ that $|E_{a,b}|\leq \Cclust \e^{-\dist(A,B)/\xi}$.
  Expanding therefore the Pfaffians in Eq.~\eqref{eq:gausscorrproofpfaffians} yields that each term either appears in both terms of the difference and cancels out or that it contains at least one element of $E$ as a factor.
  Counting the number of terms in the expansion of the Pfaffians gives that the sum contains
  \begin{equation}
  	(2[|\supp(A)|+|\supp(B)|] - 1)!!\leq2^{|\supp(A)|+|\supp(B)|}(|\supp(A)|+|\supp(B)|)^{|\supp(A)|+|\supp(B)|}
  \end{equation}
 many terms which yields the final bound stated in the Lemma.

\end{proof}

\section{Details of the proof of Theorem 1}
\label{app_proof}
In this appendix we provide all the details of the proof of our main result Theorem 1.
We proceed as follows:
In Section~\ref{app_proof_lieb_robinson_cone} we bound the error introduced by truncating to the Lieb-Robinson cone.
In Section~\ref{deltapartitions} we introduce the necessary concepts and notation to then in Section~\ref{decoupling} bound the error made by factorizing expectation values into a products of local contributions from different patches.
In Section~\ref{suppression} we use the properties of delocalizing transport to show a bound on the remaining non-Gaussian contributions to the expectation value.
Finally, in Section~\ref{fullproof}, we assemble all the parts of the proof and state a more technical version of the main theorem.

\subsection{Truncating to the Lieb-Robinson cone}
\label{app_proof_lieb_robinson_cone}
We decompose a general operator supported in the region $S$ according to Eq.~\eqref{eq:decomposeA}. 
Without loss of generality we can assume that $A$ is normalized, i.e., $\| A\|\leq 1$, which implies $|a(J)|\leq 1$ and so in the following we concentrate on the individual terms of the form given in Eq.~\eqref{eq:app_LRinitialop}.
We will demonstrate that sums over time evolved fermionic operators in Eq.~\eqref{eq:app_LRinitialop} can be truncated to an enlarged Lieb-Robinson cone up to an error that decays exponentially with time.
Rather than summing over all possible index positions $k_j \in [2 \vol]$ it is then sufficient to only sum over positions inside this enlarged Lieb-Robinson cone.
We will require the following auxiliary lemma:
\begin{lm}[Norm bound on restricted sums of fermionic operators]
  \label{lm:partial_sum}
  Let $I \subset [2 \vol]$ and $W \in U(2 \vol)$ be unitary. Then for all $j \in [2 \vol]$
  \begin{align}
    \biggl\| \sum_{k_j \in I} W_{j,k_j} c_{k_j} \biggr\| \leq 1 \; .
  \end{align}
\end{lm}
\begin{proof}
  The proof can be carried out with straightforward norm estimates and using the normalisation
  of the two-point correlator $\gamma$.
  We begin with
  \begin{align}
    \left\| \sum_{k_j \in I} W_{j,k_j} c_{k_j} \right \| 
    &= \sup_{\substack{\ket{\psi} \\ \|\ket{\psi}\|=1}} \bra{\psi} 
    \sum_{r_j \in I} \overline{W}_{j,r_j} c_{r_j}^\dagger \sum_{k_j \in I} W_{j,k_j} c_{k_j} \ket{\psi}
    = 
    \sup_{\substack{\ket{\psi} \\ \|\ket{\psi}\|=1}} 
    \sum_{r_j \in I} \sum_{k_j \in I} \overline{W}_{j,r_j}  
    \bra{\psi} c_{r_j}^\dagger c_{k_j} \ket{\psi} W_{j,k_j}\; .
  \end{align}
  We now rewrite this as a matrix multiplication on the index space 
  \begin{align}
    \left\| \sum_{k_j \in I} W_{j,k_j} c_{k_j} \right \| 
    \leq \sup_{\gamma} \bra{j} \overline{W} P_I \gamma P_I W^{T} \ket{j} \; ,
  \end{align}
  where $\ket{j}$ is a vector on the index space, $P_I$ denotes the projector onto the interval $I$
  and $\gamma$ denotes fermionic correlation matrices.
  A straightforward norm estimate and using that $\|\gamma\|\leq 1$, as every fermionic mode
  can be occupied by at most one particle, gives
  \begin{align}
    \left\| \sum_{k_j \in I} W_{j,k_j} c_{k_j} \right \| 
    \leq \| \overline{W} \| \, \| P_I \| \, \| \gamma \| \, \| P_I \| \, \| W^{T}\| \leq 1 \; ,
  \end{align}
  which concludes the proof.
\end{proof}

As introduced in the main text, we then denote by $\dist(A,B)$ the shortest distance between the supports $\supp(A),\supp(B)$ of two operators $A$ and $B$.
For $k_1,k_2 \in [2 \vol]$ we then define the distance $\dist(k_1,k_2) \coloneqq \dist(c_{k_1},c_{k_2})$. 
Note that $\dist$ defines only a pseudometric on $[2V]$ as for $k_1,k_2\in[2V]$ with $c_{k_1}=f_s$ and $c_{k_2}=f_s^\dag$ we have $k_1\neq k_2$ but $\dist(k_1,k_2)=0$.

Given a pseudometric, we define a ball around a set as follows.
\begin{dfn}[Ball around set]
 Given $l>0$, a set $M$ with pseudometric $\dist$ and $J\subset M$, we define the $l$-ball $\ball_l(J) \subset M$ around $J$ by
 \begin{align}
  \ball_l(J) = \{s\in M:\min\limits_{j\in J} \dist(j,s)\leq l\}.
 \end{align}
\end{dfn}

With this, we define an enlarged Lieb-Robinson cone around a set of indices $J$ with radius $(v+ 2 v_\epsilon) |t|$ for some $v_\epsilon >0$ and bound the error made by restricting sums of the form given in Eq.~\eqref{eq:app_LRinitialop} to this widened Lieb-Robinson cone:
\begin{lm}[Error made in restricting to widened Lieb-Robinson cone]
  \label{lm:LR}
  Given a $\dlat$-dimensional cubic lattice system with a quadratic Hamiltonian $H$ that satisfies 
  a Lieb-Robinson bound of the form given in Lemma 2 of the main text  
  with parameters $C_{\mathrm{LR}},\mu,v>0$.
  Let $v_\epsilon >0$ and define for any set $J \subset [2 \vol]$ the widened cone $\cone \coloneqq \ball_{(v+ 2 v_\epsilon)|t|}(J)$, then
  there exists a constant $\tilde{C}_{\mathrm{LR}}(\dlat)$, such that
  \begin{align}
    \biggl\| \sum_{(k_j)_{j \in J} \notin \cone^{\times |J|}}
    \prod_{j \in J} W_{j,k_j}(t)\, c_{k_j} \biggr\|
    \leq \tilde{C}_{\mathrm{LR}}(\dlat) |J|^2 \e^{-\mu v_\epsilon |t|} \; .
  \end{align}
\end{lm}
\begin{proof}
  We begin by splitting the sum according to whether the first index is inside the cone or not.
  All other indices are free if $k_{j_1}$ is outside the cone, while at least one other index is outside the cone if $k_{j_1}$ lies in it.
  Using Lemma \ref{lm:partial_sum}, we obtain
  \begin{align}
    \biggl\| \sum_{(k_j)_{j \in J} \notin \cone^{\times |J|}} 
      \prod_{j \in J} W_{j,k_j}(t)\, c_{k_j} \biggr\| 
    &\leq \biggl\| \sum_{k_{j_1} \notin \cone} W_{j_1,k_{j_1}}(t)\, c_{k_{j_1}} \biggr\| \biggl\| 
      \prod\limits_{j\in J\setminus\{j_1\}} 
      c_{j}(t) \biggr\|  \nonumber \\
    &+ \biggl\| \sum_{k_{j_1} \in \cone} W_{j_1,k_{j_1}}(t)\, c_{k_{j_1}} \biggr\| 
      \biggl\| \sum\limits_{(k_j)_{j\in J\setminus\{j_1\}}\notin \cone^{\times |J|-1}} \prod_{j\in J\setminus\{j_1\}} W_{j,k_j}(t)\, c_{k_j} \biggr\| \nonumber   \\
    &\leq \biggl\| \sum_{k_{j_1} \notin \cone} W_{j_1,k_{j_1}}(t)\, c_{k_{j_1}} \biggr\| 
      + \biggl\| \sum\limits_{(k_j)_{j\in J\setminus\{j_1\}}\notin \cone^{\times |J|-1}}  
      \prod_{j\in J\setminus\{j_1\}} W_{j,k_j}(t)\, c_{k_j} \biggr\| \; .
      \label{eq:LR_bound}
  \end{align}
  The first term in the above equation now, due to Lemma 2 of the main text , satisfies
  \begin{align}
    \biggl\| \sum_{k_{j_1} \notin \cone} W_{j_1,k_{j_1}}(t)\, c_{k_{j_1}} \biggr\|
    &\leq C_{\mathrm{LR}}\, \sum_{l=(v+2 v_\epsilon)\,|t|}^\vol |\ball_{l+1}(J) \backslash \ball_{l}(J) |\, \e^{\mu\, ( v\,|t| - l)}   \nonumber  \\
    &\leq 2^{\dlat+1} \, \dlat \, |J| \, C_{\mathrm{LR}}\, \e^{\mu\, v\, |t|} \, \sum_{l=(v+2 v_\epsilon)\,|t|}^\infty l^{\dlat-1} \, \e^{-\mu\,l} \; ,
  \end{align}
  where we have used that the number $|\ball_{l+1}(J) \backslash \ball_{l}(J) |$ of points in the surface of a cone with radius $l$ around $J$ 
  in a cubic lattice is bounded by $4\,\dlat\,|J|\,(2l)^{\dlat-1}$.
  Shifting the limits of the sum then yields
  \begin{align}
    \biggl\| \sum_{k_{j_1} \notin \cone} W_{j_1,k_{j_1}}(t)\, c_{k_{j_1}} \biggr\|
    &\leq \e^{- \mu\, v_\epsilon\, |t|}\, |J|\, 2^{\dlat+1} \, \dlat \, C_{\mathrm{LR}}\, \, 
      \sum_{l=0}^\infty \left(l + (v+2 v_\epsilon) |t| \right)^{\dlat-1} \, \e^{-\mu \, (l + v_\epsilon |t|)} \; .
  \end{align}
  We now define the time independent constant
  \begin{align}
    \tilde{C}_{\mathrm{LR}}(\dlat) \coloneqq \sup_{t \in \mathbb{R}^{+}} 2^{\dlat+1} \, \dlat \, C_{\mathrm{LR}}\, \,
      \sum_{l=0}^\infty \left(l + (v+2 v_\epsilon) |t| \right)^{\dlat-1} \, \e^{-\mu \, (l + v_\epsilon |t|)} \; 
  \end{align}
  which can be written in terms so of the Hurwitz-Lerch-Phi function $\Phi$ (also known as Lerch transcendent)
  \begin{align}
    \tilde{C}_{\mathrm{LR}}(\dlat) \coloneqq \sup_{t \in \mathbb{R}^{+}} 2^{\dlat+1} \, \dlat \, C_{\mathrm{LR}}\, \,
      \Phi(\e^{-\mu},1-\dlat,(v+2 v_\epsilon)|t| )\e^{-\mu v_\epsilon |t|} \; .
  \end{align}
  Inserting the estimate into Eq.~\eqref{eq:LR_bound} and iteratively using the resulting inequality $|J|$-times gives the result as stated.
  
  As argued above, the constant is directly related to the Hurwitz-Lerch-Phi function and can easily be explicitly evaluated for physical
  dimensions $\dlat = 1,2,3$. In one dimension, the constant takes the form
  \begin{align}
    \tilde{C}_{\mathrm{LR}} (1) = 4 \, C_{\mathrm{LR}}\, \frac{1}{1 - \e^{-\mu}} \; .
  \end{align}
\end{proof}

\subsection{Partitions: Tracking indices on the lattice}
\label{deltapartitions}
Using the result of Lemma~\ref{lm:LR} we can restrict the time evolution in Eq.~\eqref{eq:app_LRinitialop} 
to the Lieb-Robinson cone at the cost of an exponentially suppressed error term.
In this section we therefore look at
\begin{align}
  \label{eq:app_operatoroncone}
  A_J^{LR}(t) \coloneqq \sum_{\substack{(k_j)_{j\in J} \in \cone^{\times |J|}}} \quad \biggl( \prod_{j \in J} W_{j,k_j}(t)\, c_{k_j} \biggr) \; ,
\end{align}
the restriction of a term of the form Eq.~\eqref{eq:app_LRinitialop} to the widened Lieb-Robinson cone $\cone = \ball_{(v+ 2 v_\epsilon)|t|}(J)$.
By grouping summands according to how close the respective indices $k_j$ are on the lattice we will rewrite $A_J^{LR}(t)$ as a sum over partitions of the sub-index set $J$.
This will later allow us to factorize certain expectation values using the exponential decay of correlations in the initial state.

We start by introducing some notation.
Given a finite non-empty set $J$, a partition $P$ of $J$ is a set of non-empty subsets (patches) of $J$, whose union is $J$, i.e., $\overline P \coloneqq \bigcup_{p\in P}p  = J$.
We denote by $\pi_m(P) \coloneqq \{p\in P:|p|=m\}$ the subset of all patches in a partition with a given size $m$ and by $\pi_{>m}(P) = \{p\in P:|p|>m\}$ that of all patches with size larger than $m$.
We refer to patches of size two as pairs and patches of size at least four as clusters. 
Partitions will be called even, if all patches in it have an even size.
We further denote by $\partitions(J) = \{P:P \text{ partition of }J\}$ the set of all partitions of $J$, and by $\partitions_m(J)$ and $\partitions_{>m}(J)$ the sets of all partitions into patches of size exactly equal to, or larger than $m$, respectively.
Given two partitions $P,Q$ of the set $J$ we say that $Q$ is a coarsening of $P$ and write $Q > P$ if $\forall p\in P\, \exists q \in Q : p\subset q$ and $P \neq Q$.

Next we introduce the notion of a $\Delta$-partition of the sub-index set $J$.
To each configuration of indices $(k_j)_{j\in J}$, we assign a unique partition $P$ 
of the subindices $j$ such that all indices $(k_j)_{j\in p}$ with subindices that lie within one set $p$ of the partition are connected by a path of steps with maximal length $\Delta$ and all indices $k_j$ corresponding to subindices in two different sets of the partition lie more than a distance $\Delta$ apart. 
\begin{dfn}[$\Delta$-partition]
  Given a distance $\Delta > 0$, a finite set $J\subset\N$, a finite set $M$ equipped with a pseudometric $\dist$, and a sequence of elements $(k_j)_{j\in J}\in M^{\times |J|}$.
  We define the $\Delta$-partition $P_\Delta(J,(k_j)_{j\in J})$ to be the unique partition of $J$ which fulfills  
 \begin{compactenum}[(1)]
 \item Each set in the partition is path connected by hops of length at most $\Delta$ in the sense that \begin{equation} \forall p \in P_\Delta(J,(k_j)_{j\in J}): \forall x,y\in p \exists z_1,\ldots,z_N \in p: x=z_1, y=z_N \wedge \forall i\in[N-1]: \dist(k_{z_i},k_{z_{i+1}}) \leq \Delta . \end{equation}
 \item The different patches in the partition are separated by a distance larger than $\Delta$ in the sense that \begin{equation} \forall p \neq q \in P_\Delta(J,(k_j)_{j\in J}): \forall x\in p, y\in q : \dist(k_x,k_y) > \Delta . \end{equation}
   \end{compactenum}
\end{dfn}

\begin{figure}[tb] 
  \includegraphics[width=0.95\columnwidth]{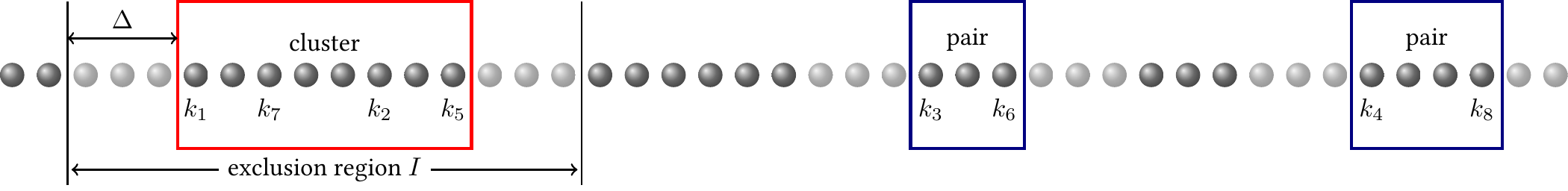}
  \caption[partition]{Illustration of a $\Delta$-partition of a given index configuration with $\Delta=3$ and $J=\{1,\ldots,8\}$. To the above configuration of the indices $k_1, \dots k_8$ we associate the $\Delta$-partition $\{\{1, 2, 5, 7\},\{3, 6\},\{4, 8\}\}$ such that a patch of size four (a cluster) and two patches of size two (pairs) are formed. Around each patch exists an buffer region (shaded nodes) which separates different patches. For a given $\Delta$-partition $P$ we obtain all possible index configurations in $\cstr{M}{P}$ by placing the patches iteratively. A placed patch will hereby create an exclusion region consisting of the patch itself and the buffer region around it in which no further patch can be placed.}
  \label{fig:partition}
\end{figure}

In addition, we define a compact notation for index configurations distributed over the lattice
such that their $\Delta$-partition agrees with a given partition $P$.
\begin{dfn}[Index sets respecting $\Delta$-partitions]
Given a distance $\Delta>0$, the set $[2 \vol]$ of all sites on the lattice equipped with a pseudometric $\dist$, a sub-index set $M \subset [2 \vol]$, and
a partition $P\in\partitions(J)$.
We denote the set of sequences contained in $M$ whose $\Delta$ partitions is equal to $P$ by
\begin{equation}
  \cstr{M}{P} \coloneqq \{(k_j)_{j \in \overline{P}} \in M^{\times |\overline{P}|}:P_\Delta(\overline{P},(k_j)_{j \in \overline{P}})=P\}.
\end{equation}
\end{dfn}
  
The notation introduced above allows us to rewrite the sum over the indices $(k_j)_{j_\in J}$ in Eq.~\eqref{eq:app_operatoroncone} inside the cone by sorting them according to their associated $\Delta$-partition.
\begin{align}
  \label{eq:app_patches}
  A_J^{LR}(t) = \sum_{\substack{(k_j)_{j\in J} \in \cone^{\times |J|}}} \quad \biggl( \prod_{j \in J} W_{j,k_j}(t)\, c_{k_j} \biggr)
  &= \sum \limits_{P\in\partitions(J)} \; \sum\limits_{(k_j)_{j\in J}\in\cstr{\cone}{P}} \;
    \prod \limits_{j\in J}W_{j,k_j}(t)\,c_{k_j}
       \nonumber\\
  &= \sum\limits_{P\in\partitions(J)} \sign(P) \sum \limits_{(k_j)_{j\in J} \in \cstr{\cone}{P}} \prod\limits_{p\in P}\patch{p},
\end{align}
where for each $p \in P$ we have introduced a patch operator $\patch{p}$, defined by
\begin{align}
  \label{eq:defpatches}
  \patch{p} = \prod\limits_{j \in p}W_{j,k_j}(t)c_{k_j} \; .
\end{align}
The $\sign(P)$ denotes the sign picked up from reordering the fermionic basis operators into the corresponding patches, where keeping the relative order of the operators
inside each patch fix.
  
\subsection{Factorizing expectation values}
\label{decoupling}
In the last section, we have developed the formalism to group indices on the lattice according to their distribution on the lattice and introduced the patch operators $\patch{p}$.
We now use the exponential clustering of correlations in the initial state to show that expectation values of products of such patch operators can be factorized into a product of expectation values of the individual patch operators up to a small error.

\begin{lm}[Factorizing expectation values in states with exponential clustering of correlations]
  \label{lm:decoupling}
  Let $\rho$ be a state that exhibits exponential clustering of correlations as defined in Definition 1 
  with system size independent parameters $\Cclust,\xi>0$.
  Let $J \subset [2 \vol]$ and $P\in\partitions(J)$ be a partition of $J$.
  Then for any distance $\Delta > 0$ it holds that
  \begin{equation}
    \sum\limits_{(k_j)_{j \in J}\in\cstr{\cone}{P}}\left| \Expv{\prod_{p \in P} \patch{p} } - \prod_{p \in P} \expv{\patch{p}}  \right| 
    \leq |J|^3\,\Cclust\, |\cone|^{|J|}\e^{-\Delta/\xi} \; .
  \end{equation}
\end{lm}
\begin{proof}
  We begin by factorizing out the contribution from the first patch $p_1 \in P$.
  Using Lemma~\ref{lm:partial_sum} and the exponential clustering of the initial state, 
  we find for a given $(k_j)_{j\in J}\in\cstr{\cone}{P}$
  \begin{equation}
    \left| \Expv{ \prod_{p \in P} \patch{p} } - \expv{\patch{p_1}} \Expv{ \prod_{p \in P\setminus\{p_1\}} \patch{p} } \right| \leq |p_1|\,|\overline{P}\setminus p_1|  \,\Cclust\, \e^{-\Delta/\xi} \; .
  \end{equation}
  Using the trivial bound $|p_1|\,|\overline{P}\setminus p_1| \leq |J|^2$, iterating the step above $|P|\leq |J|$ times and using that $|\cstr{\cone}{P}|\leq |\cone|^{|J|}$ yields the result as stated.
\end{proof}

From Lemma~\ref{lm:clusteringofcorrelationsgauss} follows that the same theorem applies to the Gaussified version of a state exhibiting clustering of correlations if we allow $\Cclust$ to scale with the size of the support according to $\Cclust \rightarrow \Cclust 4^{|J|}|J|^{|J|}$. 

\subsection{Suppression of non-Gaussian contributions}
\label{suppression}
In the last two sections we have bounded the error made in approximating expectation values of terms of the from given in Eq.~\eqref{eq:app_LRinitialop} by certain sums of products of expectation values of patch operators.
This allows us to bound the difference between the left and right hand side of
\begin{equation}
  \tr[A(t)\,\rho) - \tr(A(t)\,\rho_G] 
  \approx \sum_{J\subset \subSysInds} a(J)\, \sum\limits_{P\in\partitions(J)} \sign(P)
  \sum \limits_{(k_j)_{j \in J}\in\cstr{\cone}{P}} 
  \left ( \prod_{p \in P} \expv{\patch{p}} - \prod_{p \in P} \expvG{\patch{p}} \right ) .
\end{equation}
It is obvious that to the right hand side only partitions $P$ in which all patches $p$ are of even size can contribute, as $\rho$ and $\rho_G$ have an even particle number parity.
Moreover, as the second moments of $\rho$ and $\rho_G$ are equal by definition, the difference between the products also vanishes whenever all patches have size exactly $2$.
Hence, every contributing term contains at least one patch of size at least $4$.
In the remainder of this section we now bound the contribution of such partitions to the right hand side of the above equation.
The combinatorial nature of the problem makes this a tedious endeavor.
The final result is summarized in the following lemma, which is the last result that we need before we can assemble all the parts of the proof in Section~\ref{fullproof}.

\begin{lm}[Bounding contributions from partitions that contain large patches]
  \label{lm:bounding_clusters}
  Let $\rho$ be a state exhibiting exponential clustering of correlations as defined in Definition 1 
  with system size independent parameters $\Cclust,\xi>0$ and $\rho_G$ its Gaussified version.
  Let $\Delta \geq 1$, $J\subset[2V]$, and $P\in\partitions(J)$ be an even partition that contains 
  a patch of size at least four (cluster), and $m \coloneqq |\pi_{2}(P)|$ many patches of size two (pairs).
  Given that the time evolution of the system is governed by a Hamiltonian showing delocalizing transport as defined in Definition 2 with parameters $\Chom$, $\exphom$, for all $t\in(0,\min(\trec,V)]$ and $\sigma \in \{\rho,\rho_G\}$ it holds that
  \begin{align}
    R_{P}(J,m,t) \coloneqq{}& \left | \sum \limits_{(k_j)_{j \in J}\in\cstr{\cone}{P}} \prod_{p \in P} \expvS{\patch{p}} \right | \\
    \leq{}& (|\cone|^{1/4} \Chom t^{-\exphom} 2^{2\dlat} |J|^{\dlat} \Delta^{\dlat} )^{|J|-2m} \nonumber \\
    \times{}&\left[ (1+\Cclust |\cone|^2 \e^{-\Delta/\xi}) +
    2^{2\dlat |J|} |J|^{(\dlat + 1)|J|+1} \sum\limits_{r=0}^{|J|/2} (|\cone|^{1/4} \Chom t^{-\exphom} \Delta^{\dlat})^{2r+4} \right]^m \; . \nonumber
  \end{align}
\end{lm}
\begin{proof}
  In order to prove the above bound we separate the partition into a part containing only pairs and one containing the rest
  \begin{align}
    P &= \pi_2(P) \cup \pi_{>2}(P) \; .
  \end{align}
  For every given fixed position of the indices corresponding to clusters in $\pi_{>2}(P)$, the indices corresponding to the pairs are restricted to the set $K \coloneqq \cone\setminus\ball_\Delta(\{k_j\}_{j\in\overline{\pi_{>2}(P)}})$, as all patches are separated by a distance larger than $\Delta$.
  Thus, we can write 
  \begin{align}
    \label{eq:fullRinproof}
    R_{P}(J,m,t) = \left| \sum\limits_{\substack{(k_{j})_{j\in \overline{\pi_{>2}(P)}}\in \cstr{\cone}{\pi_{>2}(P)}}}
    \left[ \left( \prod\limits_{p\in \pi_{>2}(P)}\expvS{\patch{p}} \right)	
    \sum\limits_{\substack{(k_{j})_{j\in\overline{\pi_2(P)}}\in\cstr{K}{\pi_2(P)},\\
    K \coloneqq \cone\setminus\ball_\Delta(\{k_j\}_{j\in\overline{\pi_{>2}(P)}})}} 
    \prod\limits_{p\in\pi_2(P)}\expvS{\patch{p}} \right] \right| \; .
  \end{align}
  The homogeneous suppression due to delocalizing transport for $t\in(0,\min(\trec,V)]$ and $|\overline{\pi_{>2}(P)}| = |J|-2m$ imply
  \begin{align}
   \left\|\prod\limits_{p\in \pi_{>2}(P)}\patch{p}\right\| \leq (\Chom t^{-\exphom})^{|J|-2 m}  \; .
  \end{align}
  Using this and the triangle inequality for the first sum in Eq.~\eqref{eq:fullRinproof} we arrive at
  \begin{align}
    R_{P}(J,m,t) 
    \leq (\Chom t^{-\exphom})^{|J|-2 m} \sum\limits_{\substack{(k_{j})_{j\in \overline{\pi_{>2}(P)}}\in \cstr{\cone}{\pi_{>2}(P)}}}
    \left|\sum\limits_{\substack{(k_{j})_{j\in\overline{\pi_2(P)}}\in\cstr{K}{\pi_2(P)},\\K=\cone\setminus\ball_\Delta(\{k_j\}_{j\in\overline{\pi_{>2}(P)}})}} 
    \prod\limits_{p\in\pi_2(P)}\expvS{\patch{p}} \right| \; .
  \end{align}
  At this point, the inner sum still depends on the position of the clusters, as they create an exclusion region for the pairs (see Fig.~\ref{fig:partition}).
  Taking the supremum over such exclusion regions decouples the sum. 
  Then bounding the possible number of positions of the $|J|-2m$ indices in the $|\pi_{>2}(P)|$ clusters by
  \begin{align}
    |\cstr{\cone}{\pi_{>2}(P)}| \leq 
    |\cone|^{|\pi_{>2}P|} (2^{2\dlat} |J|^{\dlat} \Delta^{\dlat})^{|J| - 2 m} \; ,
  \end{align}
  gives
  \begin{align}
    R_{P}(J,m,t) 
    &\leq |\cone|^{|\pi_{>2}(P)|} (\Chom t^{-\exphom} 2^{2\dlat} |J|^{\dlat} \Delta^{\dlat})^{|J|-2m}\max\limits_{I\subset\cone}\left|
      \sum\limits_{\substack{(k_{j})_{j\in\overline{\pi_2(P)}}\in\cstr{\cone\setminus I}{\pi_2(P)}}} 
      \prod\limits_{p\in\pi_2(P)}\expvS{\patch{p}}	\right| \\
    &\leq(|\cone|^{1/4} \Chom t^{-\exphom} 2^{2\dlat} |J|^{\dlat} \Delta^{\dlat} )^{|J|-2m}
      \max\limits_{M\subset J: |M|=2m}\max\limits_{\substack{I\subset\cone,\\P\in\partitions_2(M)}}\left|
      \sum\limits_{\substack{(k_{j})_{j\in M}\in\cstr{\cone\setminus I}{P}}} 
      \prod\limits_{p\in P}\expvS{\patch{p}}	\right|.
  \end{align}
  We now define
  \begin{equation}
    f(m,t) \coloneqq \max\limits_{M\subset J: |M|=2m}\max\limits_{\substack{I\subset\cone,\\P\in\partitions_2(M)}}\left|
      \sum\limits_{\substack{(k_{j})_{j\in M}\in\cstr{\cone\setminus I}{P}}} 
      \prod\limits_{p\in P}\expvS{\patch{p}}	\right|
  \end{equation}
  and apply an recursive argument to achieve a bound of the form
  \begin{align}
    f(m,t) \leq C(m,t) \; .
  \end{align}
  \paragraph{Start of recursion:}
  For the case of one pair ($m=1$), by using Lemma~\ref{lm:partial_sum}, we can bound the appearing maximum as follows
  \begin{align}
      f(1,t) =
      &\max\limits_{M\subset J: |M|=2}\max\limits_{\substack{I\subset\cone}}
        \left| \sum\limits_{\substack{(k_j)_{j\in M}\in\cstr{\cone\setminus I}{\{M\}}}} \expvS{\patch{M}}	\right|
      = \max\limits_{M\subset J: |M|=2}\max\limits_{\substack{I\subset\cone}}
        \left| \sum\limits_{\substack{k_{l_1},k_{l_2}\in\cone\setminus I, \\ \{l_1,l_2\}=M: \\ \dist(k_{l_1},k_{l_2})\leq \Delta}} \expvS{\patch{M}}	\right| \\
      \leq&
        \max\limits_{M\subset J: |M|=2}\max\limits_{\substack{I\subset\cone}}\left( 
        \left|\sum\limits_{\substack{k_{l_1},k_{l_2}\in\cone\setminus I, \\ \{l_1,l_2\}=M}} \expvS{\patch{M}}	\right| 
        +\left|  \sum\limits_{\substack{k_{l_1},k_{l_2}\in\cone\setminus I, \\ \{l_1,l_2\}=M: \\ \dist(k_{l_1},k_{l_2})> \Delta}} \expvS{\patch{M}}\right| \right)
           \nonumber\\
      \leq& \max\limits_{M\subset J: |M|=2}\max\limits_{\substack{I\subset\cone}}(1 + |\cone|^2\Cclust \e^{-\Delta/\xi}) \\
      =& 1 + |\cone|^2\Cclust \e^{-\Delta/\xi}\;. % \eqqcolon C(1,t) \; .
      \label{eq:rec_start_1}
    \end{align}
    With this, we can move to setting up the recursion.
    \paragraph{Setting up the recursion:} To obtain a recursion formula for an upper bound $C(m, t)$ on $f(m,t)$ we now relax the condition that different pairs may not occupy close-by lattice regions.
  If we drop this constraint, the only remaining constraint is that paired indices $k_l$ and $k_{l^\prime}$ lie close to each other, i.e. $\dist(k_{l},k_{l^\prime})\leq\Delta$.
  For any set $M\subset J$ with $|M|=2 m$, set $I\subset\cone$, $P\in\partitions
_2(M)$ Eq.~\eqref{eq:rec_start_1} yields directly
  \begin{align}
    \left|\prod\limits_{p\in P}\sum\limits_{\substack{(k_{j})_{j\in p}\in\cstr{\cone\setminus I}{\{p\}}}}\expvS{\patch{p}}\right|
    \leq(1+\Cclust|\cone|^2\e^{-\Delta/\xi})^m.
  \end{align}
  This leaves us with controlling the difference between the constrained and unconstrained pairs and yields
  \begin{align}
    f(m,t) &\leq 
    \left (1 + |\cone|^2\Cclust \e^{-\Delta/\xi} \right)^m 
    + \max\limits_{M\subset J: |M|=2m}\max\limits_{\substack{I\subset\cone,\\P\in\partitions_2(M)}}
    g(M,P,I) , \\
    g(M,P,I) &\coloneqq 
    \left|\sum\limits_{\substack{(k_j)_{j\in M}\in\cstr{\cone\setminus I}{P}}}\prod\limits_{p\in P}\expvS{\patch{p}}
    -\prod\limits_{p\in P}\sum\limits_{\substack{(k_j)_{j\in p}\in\cstr{\cone\setminus I}{\{p\}}}}\expvS{\patch{p}}\right|	 .
  \end{align}
  In order to get the error under control, we need the notion of a coarsening of a partition defined above \cite{lint}. 
  The key insight here is that the above difference between constrained and unconstrained pairs can be captured by considering all possible coarsening of the partition $P$
  \begin{align}
    g(M,P,I) 
    &= \left|\sum\limits_{\substack{Q\in\partitions(M):\\Q>P}}\sum\limits_{\substack{(k_{j})_{j\in M}\in
      \cstr{\cone\setminus I}{Q}:\\\forall \{l,l^\prime\}\in P:\dist(k_{l},k_{l^\prime})\leq\Delta}}\prod\limits_{p\in P}\expvS{\patch{p}}\right|
    \leq 
      \sum\limits_{w=0}^{m-2} 
      \sum\limits_{\substack{Q\in\partitions(M):\\Q>P, |\pi_2(Q)|=w}} 
      \left| \sum\limits_{\substack{(k_{j})_{j\in M}\in\cstr{\cone\setminus I}{Q}:\\
      \forall \{l,l^\prime\}\in P:\dist(k_{l},k_{l^\prime})\leq\Delta}}\prod\limits_{p\in P}\expvS{\patch{p}} \right|,
  \end{align}
  where we have applied the triangle inequality and sorted the coarsenings by the numbers of pairs $w$ they have.
  This is almost the original expression $R_{P}(J,m,t)$ which this lemma is trying to bound, 
  with the exception that there is still a signature left of the fact that the clusters
  were created by joining pairs, such that consecutive indices have to be at most distance $\Delta$ apart in the cluster.
  Following the same steps as above, one can show that
  \begin{align}
    g(M,P,I) 
    &\leq
    \sum\limits_{w=0}^{m-2} 
    \sum\limits_{\substack{Q\in\partitions(M):\\Q>P, |\pi_2(Q)|=w}} 
    \left|\sum\limits_{\substack{(k_{j})_{j\in M}\in\cstr{\cone\setminus I}{Q}: \nonumber \\
    \forall \{l,l^\prime\}\in P:\dist(k_{l},k_{l^\prime})\leq\Delta}}\prod\limits_{p\in P}\expvS{\patch{p}}\right|  \nonumber \\
    &\leq
    \sum\limits_{w=0}^{m-2} 
    \sum\limits_{\substack{Q\in\partitions(M):\\Q>P, |\pi_2(Q)|=w}} 
    (\Chom t^{-\exphom})^{2m - 2w} \sum\limits_{\substack{(k_{j})_{j\in \overline{\pi_{>2}(P)}}\in \cstr{\cone}{\pi_{>2}(P)}  \nonumber\\
    \forall \{l,l^\prime\}\in P:\dist(k_{l},k_{l^\prime})\leq\Delta} }
    \left|\sum\limits_{\substack{(k_{j})_{j\in\overline{\pi_2(P)}}\in\cstr{K}{\pi_2(P)},\\K=\cone\setminus\ball_\Delta(\{k_j\}_{j\in\overline{\pi_{>2}(P)}})}} 
    \prod\limits_{p\in\pi_2(P)}\expvS{\patch{p}}	\right]  \nonumber \\
    &\leq \sum\limits_{w=0}^{m-2} 
    (2m)^{2m}
    (|\cone|^{1/4} \Chom t^{-\exphom} 2^{2\dlat} |J|^{\dlat} \Delta^{\dlat} )^{2m-2w} f(w,t)
    \; ,
  \end{align}
  where we used that for a finite set $M$, we can upper bound the number of partitions of this set by $|\partitions(M)|\leq |M|^{|M|}$.
  Thus we obtain for the function $f(m,t)$ the following upper bound
  \begin{align}
    f(m,t) 
    &\leq \left (1 + |\cone|^2\Cclust \e^{-\Delta/\xi} \right)^m 
      + \sum\limits_{w=0}^{m-2} |J|^{|J|}
      (|\cone|^{1/4} \Chom t^{-\exphom} 2^{2\dlat} |J|^{\dlat} \Delta^{\dlat} )^{2m-2w} f(w,t) \\
    &\leq \left (1 + |\cone|^2\Cclust \e^{-\Delta/\xi} \right)^m \nonumber 
    + 2^{2\dlat|J|} |J|^{ (\dlat + 1)|J|}\sum\limits_{w=0}^{m-2} 
    (|\cone|^{1/4} \Chom t^{-\exphom} \Delta^{\dlat} )^{2m-2w} f(w,t) \\
    &\leq \left (1 + |\cone|^2\Cclust \e^{-\Delta/\xi} \right)^m  \nonumber
    + 2^{2\dlat|J|} |J|^{(\dlat + 1)|J| } 
    \sum\limits_{r=0}^{|J|/2} (|\cone|^{1/4} \Chom t^{-\exphom} \Delta^{\dlat})^{2r+4}
    \sum\limits_{w=0}^{m-2} f(w,t) \\
    &\leq \alpha^m + \delta \sum\limits_{w=0}^{m-1} C(w,t) \eqqcolon C(m,t) \nonumber
    \;,
  \end{align}
  where we overestimated by introducing the sum over $r$ and adding the $m-1$ term to the second sum 
  and have introduced the abbreviations
  \begin{align}
   \alpha &= (1+\Cclust |\cone|^2 \e^{-\Delta/\xi}) ,\\
   \delta &= 2^{2\dlat |J|} |J|^{ (\dlat + 1)|J|} 
    \sum\limits_{r=0}^{|J|/2} (|\cone|^{1/4} \Chom t^{-\exphom} \Delta^{\dlat})^{2r+4} \; .
  \end{align}
  By now setting
  \begin{align}
    C(0,t) &\coloneqq f(0,t) = 1 ,\\
    C(1,t) &\coloneqq \alpha+\delta \geq \alpha \geq f(1,t) ,\\
    \forall m \geq 1: \quad C(m,t) &= \alpha^m + \delta \sum\limits_{w=0}^{m-1} C(w,t) ,
  \end{align}
  we have a recursively defined upper bound $C(m,t)$ on $f(m,t)$.
  \paragraph{Solving the recursion:}
  To resolve the recursion, we first show that 
  \begin{align}
    \alpha C(m,t) \leq C(m+1,t) \; ,
  \end{align}
  by relying on an induction. To begin with, we have
  \begin{align}
    \alpha\,C(0,t) = \alpha  \leq \alpha + \delta = C(1,t) \;.
  \end{align}
  For the induction step, we use
  \begin{align}
    \alpha\,C(m,t) &= \alpha^{m+1} + \delta \sum_{w=0}^{m-1} \alpha\,C(w,t)\\
    &\leq \alpha^{m+1} + \delta \sum_{w=0}^{m-1} C(w+1,t) \\
    &\leq \alpha^{m+1} + \delta \sum_{w=0}^{m} C(w,t) = C(m+1,t) \; ,
  \end{align}
  where we used the induction when moving from the first to the second line by relying on the 
  fact that the sum only goes until $w=m-1$ and we added $\delta C(0,t)>0$ in the last line.
  From this, we immediately know that $C(m,t)$ is monotonically increasing as a function of $m$, since
  $\alpha \geq 1$.
  This implies
  \begin{align}
    C(m,t) \leq \alpha\,C(m,t) \leq C(m+1,t) \; . 
  \end{align}
  This allows us to resolve the recursion by iteratively using this estimate as follows
  \begin{align}
    C(m,t) &= \alpha^m + \delta \sum\limits_{w=0}^{m-1} C(w,t)\\
          &\leq \alpha^m + (m-1) \delta C(m-1,t)   \nonumber\\
          &\leq \sum_{j=0}^m \alpha^{m-j} \delta^{j} \frac{m!}{(m-j)!}   \nonumber\\
          &\leq \sum_{j=0}^m \alpha^{m-j} \delta^{j} \frac{m!}{(m-j)!} \frac{|J|^j}{j!}   \nonumber\\
          &= (\alpha + |J| \delta)^m . \nonumber
  \end{align}
  For $f(m,t)$, we hence obtain
  \begin{align}
    f(m,t) &\leq (\alpha + |J| \delta)^m \\
    &= \left[ (1+\Cclust |\cone|^2 \e^{-\Delta/\xi}) +
   |J| 2^{2\dlat |J|} |J|^{(\dlat + 1)|J|} \sum\limits_{r=0}^{|J|/2} (|\cone|^{1/4} \Chom t^{-\exphom} \Delta^{\dlat})^{2r+4} \right]^m \nonumber
  \end{align}
  and for the original quantity considered in this lemma, this yields
  \begin{align}
    R_{P}(J,m,t) 
    &\leq (|\cone|^{1/4} \Chom t^{-\exphom} 2^{2\dlat} |J|^{\dlat} \Delta^{\dlat} )^{|J|-2m} \\
    &\times \left[ (1+\Cclust |\cone|^2 \e^{-\Delta/\xi}) +
   |J| 2^{2\dlat |J|} |J|^{(\dlat + 1)|J|} \sum\limits_{r=0}^{|J|/2} (|\cone|^{1/4} \Chom t^{-\exphom} \Delta^{\dlat})^{2r+4} \right]^m \; ,\nonumber
  \end{align}
  which concludes the proof.
\end{proof}

\subsection{Overview of the proof of Theorem 1}
\label{fullproof}

Collecting all the results of the preceding sub-sections, we can now proof the following result, which directly implies theorem 1 in the main text.
\begin{thm*}[Gaussification in finite time]
  \label{thm:gauss_app}
  Let $\Cclust,\xi,\Chom >0$.
  Consider a family of systems on $\dlat$-dimensional cubic lattices of increasing volume $\vol$ 
  and let $S$ be some fixed finite region of sites and $\exphom>\dlat/4$.
  Let the corresponding initial states exhibit 
  exponential clustering of correlations with constant $\Cclust$ and correlation length $\xi$.
  Let the Hamiltonians of these systems be quadratic finite range and let them exhibit delocalizing transport 
  with constants $\Chom$ and $\exphom$ and a recurrence time $\trec$
  increasing unboundedly as some function of the volume $\vol$.
  Then for any $\epsilon>0$ 
  there exists a relaxation time $\trelax > 0$ independent of the system size
  such that for all $t \in [\trelax,\,\trec]$ it holds 
  that $\| \rho^S(t) - \rho^S_G(t) \|_1 \leq \epsilon$.
\end{thm*}
\begin{proof}
  To begin with we rewrite the one-norm as
  \begin{align}
    \| \rho^S(t) - \rho^S_G(t) \|_1 
    = \sup_{A \in \mathcal{A}_S} \tr\left(A (t) (\rho - \rho_G)\right) \; ,
  \end{align}
  and expand the operator $A$ in the basis of fermionic operators
  \begin{align}
    A(t) = \sum_{b_1, \dots, b_{2 |S|}=0}^1 a_{b_1,\cdots,b_{2 |S|}} \,  c_{s_1}(t)^{b_1} \dots c_{s_{2 |S|}}{(t)}^{b_{2 |S|}} \; . 
  \end{align} 
  Normalization of the operator $\|A\|=1$ implies that all of the $2^{2 |S|}$ coefficients
  satisfy $|a_{b_1,\cdots,b_{2m}}| \leq 1$, thus
  \begin{align}
    \| \rho^S(t) - \rho^S_G(t) \|_1  
    \leq{} & 2^{2|S|} \max_{J \subset \subSysInds} \biggl|\tr\biggl( \prod_{j \in J} c_{j}(t)\, (\rho - \rho_G)\biggr)\biggr| \\
    \leq{} & 2^{2|S|} \max_{J \subset \subSysInds} \biggl| \; 
    \sum_{(k_j)_{j \in J} \in [2 \vol]^{\times |J|}} \quad
    \tr\biggl( \prod_{j \in J} W_{j,k_j}(t)\, c_{k_j} \, (\rho - \rho_G)\biggr)\biggr| \; . \nonumber
  \end{align}
  Using the Lieb-Robinson bound stated in Lemma~\ref{lm:LR}, we can restrict the sum in the right hand side of the previous inequality to the Lieb-Robinson cone $\cone$.
  This leads to an error term that is exponentially suppressed in time $t$ and we obtain
  \begin{equation}
      \| \rho^S(t) - \rho^S_G(t) \|_1 
      \leq  2^{2|S|} \max_{J \subset \subSysInds} \biggl( \biggl| \; 
    \sum_{(k_j)_{j \in J} \in [\cone]^{\times |J|}}\quad
    \tr\biggl( \prod_{j \in J} W_{j,k_j}(t)\, c_{k_j} \, (\rho - \rho_G)\biggr) \biggr|
    + 2\tilde{C}_{\mathrm{LR}}(\dlat) |J|^2 \e^{-\mu v_\epsilon |t|} \biggr) \; .
  \end{equation}
  We now reorder the terms in the sum according to how the indices $k_j$ are distributed on the lattice. 
  To that end, in Section~\ref{deltapartitions}, we have introduced the concept of a $\Delta$-partition.
  We turn the sum into a sum over all possible partitions $\partitions(J)$ and then, for each partition $P\in\partitions(J)$, sum over all possible ways $\cstr{\cone}{P}$ to distributed the indices over the lattice whose $\Delta$-partition coincides with that given partition $P$.
  Partitions consist of patches and we collect the factors $W_{j,k_j}(t)\, c_{k_j}$ from the product over $j \in J$ into patch operators $\patch{p}$ for each patch $p$, as defined in Eq.~\eqref{eq:defpatches}.
  Together with the triangle inequality this yields
  \begin{equation}
      \| \rho^S(t) - \rho^S_G(t) \|_1
      \leq 2^{2|S|} \max_{J \subset \subSysInds} \biggl( \sum\limits_{P\in\partitions(J)} \biggl| 
        \sum_{(k_j)_{j \in J} \in \cstr{\cone}{P}}\quad \biggl( \expv{\prod\limits_{p\in P}\patch{p}}-\expvG{\prod\limits_{p\in P}\patch{p}}
        \biggr) \biggr|
    + 2\tilde{C}_{\mathrm{LR}}(\dlat) |J|^2 \e^{-\mu v_\epsilon |t|} \biggr) \; .
  \end{equation}
  Lemma~\ref{lm:decoupling} allows us to factor the expectation values with respect to $\rho$ and using Lemma~\ref{lm:clusteringofcorrelationsgauss} its Gaussified version $\rho_G$ into products of expectation values of the individual patch operators.
  This leads to an additional error term that grows polynomially with the size of the cone, but is exponentially suppressed in the minimal patch distance $\Delta$, so that we get
  \begin{equation}
    \begin{split}
    \| \rho^S(t) - \rho^S_G(t) \|_1  
    \leq{} & 2^{2|S|} \max_{J \subset \subSysInds} \biggl( \sum\limits_{P\in\partitions(J)} \biggl| 
    \sum_{(k_j)_{j \in J} \in \cstr{\cone}{P}}\quad \biggl( \prod\limits_{p\in P}\expv{\patch{p}}-\prod\limits_{p\in P}\expvG{\patch{p}}
    \biggr) \biggr|
    \\ &+ (1+2^{2|J|}|J|^{|J|})|J|^{|J|+3}\Cclust|\cone|^{|J|}\e^{-\Delta/\xi} 
    + 2\tilde{C}_{\mathrm{LR}}(\dlat) |J|^2 \e^{-\mu v_\epsilon |t|} \biggr) \; .
    \end{split}
  \end{equation}
  It is now apparent that partitions that contain at least one patch of odd size do not contribute to the sum as then the corresponding patch operator does not fulfill the parity super-selection rule.
  Likewise, partitions that contain only patches of size two do not contribute, as the expectation values of their patch operators are the same in $\rho$ and $\rho_G$.
  It remains to bound the contribution from the remaining partitions.
  For these we cannot use cancellations between the parts coming from $\rho$ and those coming from $\rho_G$, but instead bound them in absolute value.
  All these partitions contain at least one cluster of size at least four, which allows us to bound the corresponding term from the homogeneous suppression of the elements of the propagator implied by the delocalizing transport (see Definition 2).
  Doing this explicitly is tedious because of the interplay of contributions from the larger patches and those of size two, and the involved combinatorics of how the smaller patches can be distributed on the lattice.
  All this is done by first ordering contributions according to the number $m$ of patches of size two they contain and then applying Lemma~\ref{lm:bounding_clusters}, which internally uses a recursive argument.
  It yields an upper bound on the absolute value of sums of the form $\sum_{(k_j)_{j \in J} \in \cstr{\cone}{P}} \prod\limits_{p\in P} \expv{\patch{p}}$ that grows with $\Delta$ but is algebraically suppressed with time $t$.
  Assuming $\Delta \geq 1$ this yields for $t\leq\min(\trec,V)$ the following bound
  \begin{equation}
    \begin{split}
    \| \rho^S(t) - \rho^S_G(t) \|_1  
    \leq{} & 2^{2|S|} \biggl[2^{2|S|+1}|S|^{2|S|}\sum\limits_{m=0}^{|S|}
    (2^{3\dlat}|S|^{\dlat}|\cone|^{1/4}\Chom t^{-\exphom}\Delta^{\dlat})^{2m+4}
    \biggl(1+\Cclust|\cone|^2\e^{-\Delta/\xi}\\
    &+2^{2(3\dlat+1)|S|+1}|S|^{2|S|(\dlat+1)+1}\sum\limits_{r=0}^{|S|}(|\cone|^{1/4}\Chom t^{-\exphom}\Delta^{\dlat})^{2r+4}\biggr)^{|S|}
    \\ &+ 8 \tilde{C}_{\mathrm{LR}}(\dlat) |S|^2 \e^{-\mu v_\epsilon |t|}
      + 2^{8|S|+4}|S|^{4|S|+3}\Cclust|\cone|^{2|S|}\e^{-\Delta/\xi} \biggr] \; .
    \end{split}
    \label{eq:final_bound}
  \end{equation}
  Recalling that $|\cone|\leq2|S|[2(v+2v_\epsilon)t+1]^{\dlat}$ and setting $v_\epsilon=v$ one realizes that by letting $\Delta$ grow in a suitable way with $t$, all terms are at least algebraically suppressed in $t$.
  More precisely, this happens for all $\Delta = \max(1,t^{\nu/4\dlat})$ with $0 < \nu < 4\exphom - \dlat$.
  For large enough times the bound is then dominated by a power-law originating from the $m=0$ term of the first sum.
  Eq.~\eqref{eq:final_bound} in particular implies that there exists a constant $C_{\mathrm{Total}}$ that depends on $\nu$, $|S|$, and the physical parameters of the Hamiltonian and initial state, but that is independent of the size of the total system, such that for all $t\leq\min(\trec,V)$
  \begin{equation}
    \| \rho^S(t) - \rho^S_G(t) \|_1 \leq C_{\mathrm{Total}} \, t^{-4\exphom + \dlat + \nu} .
  \end{equation}
  Moreover, for every $\epsilon$ there exists a critical system size from which on the bound above becomes smaller than that $\epsilon$ for a suitable relaxation time $\trelax\leq\min(\trec,V)$. The saturation of the delocalizing transport once $t$ is equal to $V$ implies a minimal value that bound can achieve for any given $V$. This sets the minimal system size for which $\trelax<\trec$.
\end{proof}
\end{widetext}


\begin{thebibliography}{36}%
\makeatletter
\providecommand \@ifxundefined [1]{%
 \@ifx{#1\undefined}
}%
\providecommand \@ifnum [1]{%
 \ifnum #1\expandafter \@firstoftwo
 \else \expandafter \@secondoftwo
 \fi
}%
\providecommand \@ifx [1]{%
 \ifx #1\expandafter \@firstoftwo
 \else \expandafter \@secondoftwo
 \fi
}%
\providecommand \natexlab [1]{#1}%
\providecommand \enquote  [1]{``#1''}%
\providecommand \bibnamefont  [1]{#1}%
\providecommand \bibfnamefont [1]{#1}%
\providecommand \citenamefont [1]{#1}%
\providecommand \href@noop [0]{\@secondoftwo}%
\providecommand \href [0]{\begingroup \@sanitize@url \@href}%
\providecommand \@href[1]{\@@startlink{#1}\@@href}%
\providecommand \@@href[1]{\endgroup#1\@@endlink}%
\providecommand \@sanitize@url [0]{\catcode `\\12\catcode `\$12\catcode
  `\&12\catcode `\#12\catcode `\^12\catcode `\_12\catcode `\%12\relax}%
\providecommand \@@startlink[1]{}%
\providecommand \@@endlink[0]{}%
\providecommand \url  [0]{\begingroup\@sanitize@url \@url }%
\providecommand \@url [1]{\endgroup\@href {#1}{\urlprefix }}%
\providecommand \urlprefix  [0]{URL }%
\providecommand \Eprint [0]{\href }%
\providecommand \doibase [0]{http://dx.doi.org/}%
\providecommand \selectlanguage [0]{\@gobble}%
\providecommand \bibinfo  [0]{\@secondoftwo}%
\providecommand \bibfield  [0]{\@secondoftwo}%
\providecommand \translation [1]{[#1]}%
\providecommand \BibitemOpen [0]{}%
\providecommand \bibitemStop [0]{}%
\providecommand \bibitemNoStop [0]{.\EOS\space}%
\providecommand \EOS [0]{\spacefactor3000\relax}%
\providecommand \BibitemShut  [1]{\csname bibitem#1\endcsname}%
\let\auto@bib@innerbib\@empty
%</preamble>
\bibitem [{\citenamefont {Polkovnikov}\ \emph {et~al.}(2011)\citenamefont
  {Polkovnikov}, \citenamefont {Sengupta}, \citenamefont {Silva},\ and\
  \citenamefont {Vengalattore}}]{PolkovnikovReview}%
  \BibitemOpen
  \bibfield  {author} {\bibinfo {author} {\bibfnamefont {A.}~\bibnamefont
  {Polkovnikov}}, \bibinfo {author} {\bibfnamefont {K.}~\bibnamefont
  {Sengupta}}, \bibinfo {author} {\bibfnamefont {A.}~\bibnamefont {Silva}}, \
  and\ \bibinfo {author} {\bibfnamefont {M.}~\bibnamefont {Vengalattore}},\
  }\href@noop {} {\bibfield  {journal} {\bibinfo  {journal} {Rev. Mod. Phys.}\
  }\textbf {\bibinfo {volume} {83}},\ \bibinfo {pages} {863} (\bibinfo {year}
  {2011})}\BibitemShut {NoStop}%
\bibitem [{\citenamefont {Eisert}\ \emph {et~al.}(2015)\citenamefont {Eisert},
  \citenamefont {Friesdorf},\ and\ \citenamefont {Gogolin}}]{1408.5148}%
  \BibitemOpen
  \bibfield  {author} {\bibinfo {author} {\bibfnamefont {J.}~\bibnamefont
  {Eisert}}, \bibinfo {author} {\bibfnamefont {M.}~\bibnamefont {Friesdorf}}, \
  and\ \bibinfo {author} {\bibfnamefont {C.}~\bibnamefont {Gogolin}},\
  }\href@noop {} {\bibfield  {journal} {\bibinfo  {journal} {Nature Phys.}\
  }\textbf {\bibinfo {volume} {11}},\ \bibinfo {pages} {124} (\bibinfo {year}
  {2015})}\BibitemShut {NoStop}%
\bibitem [{\citenamefont {Gogolin}\ and\ \citenamefont
  {Eisert}(2016)}]{1503.07538}%
  \BibitemOpen
  \bibfield  {author} {\bibinfo {author} {\bibfnamefont {C.}~\bibnamefont
  {Gogolin}}\ and\ \bibinfo {author} {\bibfnamefont {J.}~\bibnamefont
  {Eisert}},\ }\href@noop {} {\bibfield  {journal} {\bibinfo  {journal} {Rep.
  Prog. Phys.}\ }\textbf {\bibinfo {volume} {79}},\ \bibinfo {pages} {056001}
  (\bibinfo {year} {2016})}\BibitemShut {NoStop}%
\bibitem [{\citenamefont {Bloch}\ \emph {et~al.}(2012)\citenamefont {Bloch},
  \citenamefont {Dalibard},\ and\ \citenamefont {Nascimbene}}]{BlochSimulator}%
  \BibitemOpen
  \bibfield  {author} {\bibinfo {author} {\bibfnamefont {I.}~\bibnamefont
  {Bloch}}, \bibinfo {author} {\bibfnamefont {J.}~\bibnamefont {Dalibard}}, \
  and\ \bibinfo {author} {\bibfnamefont {S.}~\bibnamefont {Nascimbene}},\
  }\href@noop {} {\bibfield  {journal} {\bibinfo  {journal} {Nature Phys.}\
  }\textbf {\bibinfo {volume} {8}},\ \bibinfo {pages} {267} (\bibinfo {year}
  {2012})}\BibitemShut {NoStop}%
\bibitem [{\citenamefont {Trotzky}\ \emph {et~al.}(2012)\citenamefont
  {Trotzky}, \citenamefont {Chen}, \citenamefont {Flesch}, \citenamefont
  {McCulloch}, \citenamefont {Schollw\"ock}, \citenamefont {Eisert},\ and\
  \citenamefont {Bloch}}]{nature_bloch_eisert}%
  \BibitemOpen
  \bibfield  {author} {\bibinfo {author} {\bibfnamefont {S.}~\bibnamefont
  {Trotzky}}, \bibinfo {author} {\bibfnamefont {Y.-O.}\ \bibnamefont {Chen}},
  \bibinfo {author} {\bibfnamefont {A.}~\bibnamefont {Flesch}}, \bibinfo
  {author} {\bibfnamefont {I.}~\bibnamefont {McCulloch}}, \bibinfo {author}
  {\bibfnamefont {U.}~\bibnamefont {Schollw\"ock}}, \bibinfo {author}
  {\bibfnamefont {J.}~\bibnamefont {Eisert}}, \ and\ \bibinfo {author}
  {\bibfnamefont {I.}~\bibnamefont {Bloch}},\ }\href@noop {} {\bibfield
  {journal} {\bibinfo  {journal} {Nature Phys.}\ }\textbf {\bibinfo {volume}
  {8}},\ \bibinfo {pages} {325} (\bibinfo {year} {2012})}\BibitemShut {NoStop}%
\bibitem [{\citenamefont {Cramer}\ \emph {et~al.}(2008)\citenamefont {Cramer},
  \citenamefont {Dawson}, \citenamefont {Eisert},\ and\ \citenamefont
  {Osborne}}]{CramerEisert}%
  \BibitemOpen
  \bibfield  {author} {\bibinfo {author} {\bibfnamefont {M.}~\bibnamefont
  {Cramer}}, \bibinfo {author} {\bibfnamefont {C.~M.}\ \bibnamefont {Dawson}},
  \bibinfo {author} {\bibfnamefont {J.}~\bibnamefont {Eisert}}, \ and\ \bibinfo
  {author} {\bibfnamefont {T.~J.}\ \bibnamefont {Osborne}},\ }\href@noop {}
  {\bibfield  {journal} {\bibinfo  {journal} {Phys. Rev. Lett.}\ }\textbf
  {\bibinfo {volume} {100}},\ \bibinfo {pages} {030602} (\bibinfo {year}
  {2008})}\BibitemShut {NoStop}%
\bibitem [{\citenamefont {Linden}\ \emph {et~al.}(2009)\citenamefont {Linden},
  \citenamefont {Popescu}, \citenamefont {Short},\ and\ \citenamefont
  {Winter}}]{Linden_etal09}%
  \BibitemOpen
  \bibfield  {author} {\bibinfo {author} {\bibfnamefont {N.}~\bibnamefont
  {Linden}}, \bibinfo {author} {\bibfnamefont {S.}~\bibnamefont {Popescu}},
  \bibinfo {author} {\bibfnamefont {A.~J.}\ \bibnamefont {Short}}, \ and\
  \bibinfo {author} {\bibfnamefont {A.}~\bibnamefont {Winter}},\ }\href@noop {}
  {\bibfield  {journal} {\bibinfo  {journal} {Phys. Rev. E}\ }\textbf {\bibinfo
  {volume} {79}},\ \bibinfo {pages} {061103} (\bibinfo {year}
  {2009})}\BibitemShut {NoStop}%
\bibitem [{\citenamefont {Short}\ and\ \citenamefont
  {Farrelly}(2012)}]{1110.5759}%
  \BibitemOpen
  \bibfield  {author} {\bibinfo {author} {\bibfnamefont {A.~J.}\ \bibnamefont
  {Short}}\ and\ \bibinfo {author} {\bibfnamefont {T.~C.}\ \bibnamefont
  {Farrelly}},\ }\href@noop {} {\bibfield  {journal} {\bibinfo  {journal} {New
  J. Phys.}\ }\textbf {\bibinfo {volume} {14}},\ \bibinfo {pages} {013063}
  (\bibinfo {year} {2012})}\BibitemShut {NoStop}%
\bibitem [{\citenamefont {Rigol}\ \emph {et~al.}(2007)\citenamefont {Rigol},
  \citenamefont {Dunjko}, \citenamefont {Yurovsky},\ and\ \citenamefont
  {Olshanii}}]{RigolFirst}%
  \BibitemOpen
  \bibfield  {author} {\bibinfo {author} {\bibfnamefont {M.}~\bibnamefont
  {Rigol}}, \bibinfo {author} {\bibfnamefont {V.}~\bibnamefont {Dunjko}},
  \bibinfo {author} {\bibfnamefont {V.}~\bibnamefont {Yurovsky}}, \ and\
  \bibinfo {author} {\bibfnamefont {M.}~\bibnamefont {Olshanii}},\ }\href@noop
  {} {\bibfield  {journal} {\bibinfo  {journal} {Phys. Rev. Lett.}\ }\textbf
  {\bibinfo {volume} {98}},\ \bibinfo {pages} {050405} (\bibinfo {year}
  {2007})}\BibitemShut {NoStop}%
\bibitem [{\citenamefont {Reimann}\ and\ \citenamefont
  {Kastner}(2012)}]{ReimannKastner12}%
  \BibitemOpen
  \bibfield  {author} {\bibinfo {author} {\bibfnamefont {P.}~\bibnamefont
  {Reimann}}\ and\ \bibinfo {author} {\bibfnamefont {M.}~\bibnamefont
  {Kastner}},\ }\href@noop {} {\bibfield  {journal} {\bibinfo  {journal} {New
  J. Phys.}\ }\textbf {\bibinfo {volume} {14}},\ \bibinfo {pages} {043020}
  (\bibinfo {year} {2012})}\BibitemShut {NoStop}%
\bibitem [{\citenamefont {Malabarba}\ \emph {et~al.}(2014)\citenamefont
  {Malabarba}, \citenamefont {Garcia-Pintos}, \citenamefont {Linden},
  \citenamefont {Farrelly},\ and\ \citenamefont {Short}}]{PhysRevE.90.012121}%
  \BibitemOpen
  \bibfield  {author} {\bibinfo {author} {\bibfnamefont {A.~S.~L.}\
  \bibnamefont {Malabarba}}, \bibinfo {author} {\bibfnamefont {L.~P.}\
  \bibnamefont {Garcia-Pintos}}, \bibinfo {author} {\bibfnamefont
  {N.}~\bibnamefont {Linden}}, \bibinfo {author} {\bibfnamefont {T.~C.}\
  \bibnamefont {Farrelly}}, \ and\ \bibinfo {author} {\bibfnamefont {A.~J.}\
  \bibnamefont {Short}},\ }\href@noop {} {\bibfield  {journal} {\bibinfo
  {journal} {Phys. Rev. E}\ }\textbf {\bibinfo {volume} {90}},\ \bibinfo
  {pages} {012121} (\bibinfo {year} {2014})}\BibitemShut {NoStop}%
\bibitem [{\citenamefont {Pertot}\ \emph {et~al.}(2014)\citenamefont {Pertot},
  \citenamefont {Sheikhan}, \citenamefont {Cocchi}, \citenamefont {Miller},
  \citenamefont {Bohn}, \citenamefont {Koschorreck}, \citenamefont {K{\"o}hl},\
  and\ \citenamefont {Kollath}}]{Koehl3}%
  \BibitemOpen
  \bibfield  {author} {\bibinfo {author} {\bibfnamefont {D.}~\bibnamefont
  {Pertot}}, \bibinfo {author} {\bibfnamefont {A.}~\bibnamefont {Sheikhan}},
  \bibinfo {author} {\bibfnamefont {E.}~\bibnamefont {Cocchi}}, \bibinfo
  {author} {\bibfnamefont {L.~A.}\ \bibnamefont {Miller}}, \bibinfo {author}
  {\bibfnamefont {J.~E.}\ \bibnamefont {Bohn}}, \bibinfo {author}
  {\bibfnamefont {M.}~\bibnamefont {Koschorreck}}, \bibinfo {author}
  {\bibfnamefont {M.}~\bibnamefont {K{\"o}hl}}, \ and\ \bibinfo {author}
  {\bibfnamefont {C.}~\bibnamefont {Kollath}},\ }\href@noop {} {\bibfield
  {journal} {\bibinfo  {journal} {Phys. Rev. Lett.}\ }\textbf {\bibinfo
  {volume} {113}},\ \bibinfo {pages} {170403} (\bibinfo {year}
  {2014})}\BibitemShut {NoStop}%
\bibitem [{\citenamefont {Hastings}\ and\ \citenamefont
  {Koma}(2006)}]{math-ph/0507008}%
  \BibitemOpen
  \bibfield  {author} {\bibinfo {author} {\bibfnamefont {M.~B.}\ \bibnamefont
  {Hastings}}\ and\ \bibinfo {author} {\bibfnamefont {T.}~\bibnamefont
  {Koma}},\ }\href@noop {} {\bibfield  {journal} {\bibinfo  {journal} {Commun.
  Math. Phys.}\ }\textbf {\bibinfo {volume} {265}},\ \bibinfo {pages} {781}
  (\bibinfo {year} {2006})}\BibitemShut {NoStop}%
\bibitem [{\citenamefont {Nachtergaele}\ and\ \citenamefont
  {Sims}(2007)}]{Nachtergaele2013}%
  \BibitemOpen
  \bibfield  {author} {\bibinfo {author} {\bibfnamefont {B.}~\bibnamefont
  {Nachtergaele}}\ and\ \bibinfo {author} {\bibfnamefont {R.}~\bibnamefont
  {Sims}},\ }\href@noop {} {\bibfield  {journal} {\bibinfo  {journal} {Selected
  contributions of the XVth International Congress on Mathematical Physics}\
  }\bibinfo {series} {New Trends in Mathematical Physics} (\bibinfo {year}
  {2007})}\BibitemShut {NoStop}%
\bibitem [{\citenamefont {Kliesch}\ \emph {et~al.}(2014)\citenamefont
  {Kliesch}, \citenamefont {Gogolin}, \citenamefont {Kastoryano}, \citenamefont
  {Riera},\ and\ \citenamefont {Eisert}}]{Kliesch2014}%
  \BibitemOpen
  \bibfield  {author} {\bibinfo {author} {\bibfnamefont {M.}~\bibnamefont
  {Kliesch}}, \bibinfo {author} {\bibfnamefont {C.}~\bibnamefont {Gogolin}},
  \bibinfo {author} {\bibfnamefont {M.~J.}\ \bibnamefont {Kastoryano}},
  \bibinfo {author} {\bibfnamefont {A.}~\bibnamefont {Riera}}, \ and\ \bibinfo
  {author} {\bibfnamefont {J.}~\bibnamefont {Eisert}},\ }\href@noop {}
  {\bibfield  {journal} {\bibinfo  {journal} {Phys. Rev. X}\ }\textbf {\bibinfo
  {volume} {4}},\ \bibinfo {pages} {031019} (\bibinfo {year}
  {2014})}\BibitemShut {NoStop}%
\bibitem [{\citenamefont {Hudson}(1973)}]{Hudson}%
  \BibitemOpen
  \bibfield  {author} {\bibinfo {author} {\bibfnamefont {R.~L.}\ \bibnamefont
  {Hudson}},\ }\href@noop {} {\bibfield  {journal} {\bibinfo  {journal} {J.
  Appl. Prob.}\ }\textbf {\bibinfo {volume} {10}},\ \bibinfo {pages} {502}
  (\bibinfo {year} {1973})}\BibitemShut {NoStop}%
\bibitem [{\citenamefont {Dudnikova}\ \emph {et~al.}(2003)\citenamefont
  {Dudnikova}, \citenamefont {Komech},\ and\ \citenamefont {Spohn}}]{Spohn}%
  \BibitemOpen
  \bibfield  {author} {\bibinfo {author} {\bibfnamefont {T.~V.}\ \bibnamefont
  {Dudnikova}}, \bibinfo {author} {\bibfnamefont {A.}~\bibnamefont {Komech}}, \
  and\ \bibinfo {author} {\bibfnamefont {H.}~\bibnamefont {Spohn}},\
  }\href@noop {} {\bibfield  {journal} {\bibinfo  {journal} {J. Math. Phys.}\
  }\textbf {\bibinfo {volume} {44}},\ \bibinfo {pages} {2596} (\bibinfo {year}
  {2003})}\BibitemShut {NoStop}%
\bibitem [{\citenamefont {Lieb}\ and\ \citenamefont
  {Robinson}(1972)}]{liebrobinson}%
  \BibitemOpen
  \bibfield  {author} {\bibinfo {author} {\bibfnamefont {E.~H.}\ \bibnamefont
  {Lieb}}\ and\ \bibinfo {author} {\bibfnamefont {D.~W.}\ \bibnamefont
  {Robinson}},\ }\href@noop {} {\bibfield  {journal} {\bibinfo  {journal}
  {Commun. Math. Phys.}\ }\textbf {\bibinfo {volume} {28}},\ \bibinfo {pages}
  {251} (\bibinfo {year} {1972})}\BibitemShut {NoStop}%
\bibitem [{\citenamefont {Hastings}(2004)}]{Hastings2004a}%
  \BibitemOpen
  \bibfield  {author} {\bibinfo {author} {\bibfnamefont {M.~B.}\ \bibnamefont
  {Hastings}},\ }\href@noop {} {\bibfield  {journal} {\bibinfo  {journal}
  {Phys. Rev. Lett.}\ }\textbf {\bibinfo {volume} {93}},\ \bibinfo {pages}
  {126402} (\bibinfo {year} {2004})}\BibitemShut {NoStop}%
\bibitem [{\citenamefont {Flesch}\ \emph {et~al.}(2008)\citenamefont {Flesch},
  \citenamefont {Cramer}, \citenamefont {McCulloch}, \citenamefont
  {Scholl\-w\"ock},\ and\ \citenamefont {Eisert}}]{CramerEisertScholl08}%
  \BibitemOpen
  \bibfield  {author} {\bibinfo {author} {\bibfnamefont {A.}~\bibnamefont
  {Flesch}}, \bibinfo {author} {\bibfnamefont {M.}~\bibnamefont {Cramer}},
  \bibinfo {author} {\bibfnamefont {I.~P.}\ \bibnamefont {McCulloch}}, \bibinfo
  {author} {\bibfnamefont {U.}~\bibnamefont {Scholl\-w\"ock}}, \ and\ \bibinfo
  {author} {\bibfnamefont {J.}~\bibnamefont {Eisert}},\ }\href@noop {}
  {\bibfield  {journal} {\bibinfo  {journal} {Phys. Rev. A}\ }\textbf {\bibinfo
  {volume} {78}},\ \bibinfo {pages} {033608} (\bibinfo {year}
  {2008})}\BibitemShut {NoStop}%
\bibitem [{\citenamefont {Calabrese}\ and\ \citenamefont
  {Cardy}(2006)}]{CalabreseCardy06}%
  \BibitemOpen
  \bibfield  {author} {\bibinfo {author} {\bibfnamefont {P.}~\bibnamefont
  {Calabrese}}\ and\ \bibinfo {author} {\bibfnamefont {J.}~\bibnamefont
  {Cardy}},\ }\href@noop {} {\bibfield  {journal} {\bibinfo  {journal} {Phys.
  Rev. Lett.}\ }\textbf {\bibinfo {volume} {96}},\ \bibinfo {pages} {136801}
  (\bibinfo {year} {2006})}\BibitemShut {NoStop}%
\bibitem [{\citenamefont {Calabrese}\ \emph {et~al.}(2012)\citenamefont
  {Calabrese}, \citenamefont {Essler},\ and\ \citenamefont
  {Fagotti}}]{1205.2211}%
  \BibitemOpen
  \bibfield  {author} {\bibinfo {author} {\bibfnamefont {P.}~\bibnamefont
  {Calabrese}}, \bibinfo {author} {\bibfnamefont {F.~H.}\ \bibnamefont
  {Essler}}, \ and\ \bibinfo {author} {\bibfnamefont {M.}~\bibnamefont
  {Fagotti}},\ }\href@noop {} {\bibfield  {journal} {\bibinfo  {journal} {J.
  Stat. Mech.}\ ,\ \bibinfo {pages} {P07022\,}} (\bibinfo {year}
  {2012})}\BibitemShut {NoStop}%
\bibitem [{\citenamefont {Peschel}\ and\ \citenamefont
  {Eisler}(2009)}]{0906.1663}%
  \BibitemOpen
  \bibfield  {author} {\bibinfo {author} {\bibfnamefont {I.}~\bibnamefont
  {Peschel}}\ and\ \bibinfo {author} {\bibfnamefont {V.}~\bibnamefont
  {Eisler}},\ }\href@noop {} {\bibfield  {journal} {\bibinfo  {journal} {J.
  Phys. A}\ }\textbf {\bibinfo {volume} {42}},\ \bibinfo {pages} {504003}
  (\bibinfo {year} {2009})}\BibitemShut {NoStop}%
\bibitem [{\citenamefont {Fagotti}\ and\ \citenamefont
  {Essler}(2013)}]{1302.6944}%
  \BibitemOpen
  \bibfield  {author} {\bibinfo {author} {\bibfnamefont {M.}~\bibnamefont
  {Fagotti}}\ and\ \bibinfo {author} {\bibfnamefont {F.~H.~L.}\ \bibnamefont
  {Essler}},\ }\href@noop {} {\bibfield  {journal} {\bibinfo  {journal} {Phys.
  Rev. B}\ }\textbf {\bibinfo {volume} {87}},\ \bibinfo {pages} {245107}
  (\bibinfo {year} {2013})}\BibitemShut {NoStop}%
\bibitem [{\citenamefont {Bhattacharyya}\ \emph {et~al.}(2012)\citenamefont
  {Bhattacharyya}, \citenamefont {Das},\ and\ \citenamefont
  {Dasgupta}}]{PeriodicDriven_2}%
  \BibitemOpen
  \bibfield  {author} {\bibinfo {author} {\bibfnamefont {S.}~\bibnamefont
  {Bhattacharyya}}, \bibinfo {author} {\bibfnamefont {A.}~\bibnamefont {Das}},
  \ and\ \bibinfo {author} {\bibfnamefont {S.}~\bibnamefont {Dasgupta}},\
  }\href@noop {} {\bibfield  {journal} {\bibinfo  {journal} {Phys. Rev. B}\
  }\textbf {\bibinfo {volume} {86}},\ \bibinfo {pages} {054410} (\bibinfo
  {year} {2012})}\BibitemShut {NoStop}%
\bibitem [{\citenamefont {Manmana}\ \emph {et~al.}(2007)\citenamefont
  {Manmana}, \citenamefont {Wessel}, \citenamefont {Noack},\ and\ \citenamefont
  {Muramatsu}}]{Muramatsu}%
  \BibitemOpen
  \bibfield  {author} {\bibinfo {author} {\bibfnamefont {S.~R.}\ \bibnamefont
  {Manmana}}, \bibinfo {author} {\bibfnamefont {S.}~\bibnamefont {Wessel}},
  \bibinfo {author} {\bibfnamefont {R.~M.}\ \bibnamefont {Noack}}, \ and\
  \bibinfo {author} {\bibfnamefont {A.}~\bibnamefont {Muramatsu}},\ }\href@noop
  {} {\bibfield  {journal} {\bibinfo  {journal} {Phys. Rev. Lett.}\ }\textbf
  {\bibinfo {volume} {98}},\ \bibinfo {pages} {210405} (\bibinfo {year}
  {2007})}\BibitemShut {NoStop}%
\bibitem [{\citenamefont {Schneider}\ \emph {et~al.}(2012)\citenamefont
  {Schneider}, \citenamefont {Hackermuller}, \citenamefont {Ronzheimer},
  \citenamefont {Will}, \citenamefont {Braun}, \citenamefont {Best},
  \citenamefont {Bloch}, \citenamefont {Demler}, \citenamefont {Mandt},
  \citenamefont {Rasch},\ and\ \citenamefont
  {Rosch}}]{Schneider_fermionic_transport}%
  \BibitemOpen
  \bibfield  {author} {\bibinfo {author} {\bibfnamefont {U.}~\bibnamefont
  {Schneider}}, \bibinfo {author} {\bibfnamefont {L.}~\bibnamefont
  {Hackermuller}}, \bibinfo {author} {\bibfnamefont {J.~P.}\ \bibnamefont
  {Ronzheimer}}, \bibinfo {author} {\bibfnamefont {S.}~\bibnamefont {Will}},
  \bibinfo {author} {\bibfnamefont {S.}~\bibnamefont {Braun}}, \bibinfo
  {author} {\bibfnamefont {T.}~\bibnamefont {Best}}, \bibinfo {author}
  {\bibfnamefont {I.}~\bibnamefont {Bloch}}, \bibinfo {author} {\bibfnamefont
  {E.}~\bibnamefont {Demler}}, \bibinfo {author} {\bibfnamefont
  {S.}~\bibnamefont {Mandt}}, \bibinfo {author} {\bibfnamefont
  {D.}~\bibnamefont {Rasch}}, \ and\ \bibinfo {author} {\bibfnamefont
  {A.}~\bibnamefont {Rosch}},\ }\href@noop {} {\bibfield  {journal} {\bibinfo
  {journal} {Nature Phys.}\ }\textbf {\bibinfo {volume} {8}},\ \bibinfo {pages}
  {213} (\bibinfo {year} {2012})}\BibitemShut {NoStop}%
\bibitem [{\citenamefont {Parsons}\ \emph {et~al.}(2015)\citenamefont
  {Parsons}, \citenamefont {Huber}, \citenamefont {Mazurenko}, \citenamefont
  {Chiu}, \citenamefont {Setiawan}, \citenamefont {Wooley-Brown}, \citenamefont
  {Blatt},\ and\ \citenamefont {Greiner}}]{GreinerFermions}%
  \BibitemOpen
  \bibfield  {author} {\bibinfo {author} {\bibfnamefont {M.~F.}\ \bibnamefont
  {Parsons}}, \bibinfo {author} {\bibfnamefont {F.}~\bibnamefont {Huber}},
  \bibinfo {author} {\bibfnamefont {A.}~\bibnamefont {Mazurenko}}, \bibinfo
  {author} {\bibfnamefont {C.~S.}\ \bibnamefont {Chiu}}, \bibinfo {author}
  {\bibfnamefont {W.}~\bibnamefont {Setiawan}}, \bibinfo {author}
  {\bibfnamefont {K.}~\bibnamefont {Wooley-Brown}}, \bibinfo {author}
  {\bibfnamefont {S.}~\bibnamefont {Blatt}}, \ and\ \bibinfo {author}
  {\bibfnamefont {M.}~\bibnamefont {Greiner}},\ }\href@noop {} {\bibfield
  {journal} {\bibinfo  {journal} {Phys. Rev. Lett.}\ }\textbf {\bibinfo
  {volume} {114}},\ \bibinfo {pages} {213002} (\bibinfo {year}
  {2015})}\BibitemShut {NoStop}%
\bibitem [{\citenamefont {Strohmaier}\ \emph {et~al.}(2007)\citenamefont
  {Strohmaier}, \citenamefont {Takasu}, \citenamefont {G{\"u}nther},
  \citenamefont {J{\"o}rdens}, \citenamefont {K{\"o}hl}, \citenamefont
  {Moritz},\ and\ \citenamefont {Esslinger}}]{Koehl2}%
  \BibitemOpen
  \bibfield  {author} {\bibinfo {author} {\bibfnamefont {N.}~\bibnamefont
  {Strohmaier}}, \bibinfo {author} {\bibfnamefont {Y.}~\bibnamefont {Takasu}},
  \bibinfo {author} {\bibfnamefont {K.}~\bibnamefont {G{\"u}nther}}, \bibinfo
  {author} {\bibfnamefont {R.}~\bibnamefont {J{\"o}rdens}}, \bibinfo {author}
  {\bibfnamefont {M.}~\bibnamefont {K{\"o}hl}}, \bibinfo {author}
  {\bibfnamefont {H.}~\bibnamefont {Moritz}}, \ and\ \bibinfo {author}
  {\bibfnamefont {T.}~\bibnamefont {Esslinger}},\ }\href@noop {} {\bibfield
  {journal} {\bibinfo  {journal} {Phys. Rev. Lett.}\ }\textbf {\bibinfo
  {volume} {99}},\ \bibinfo {pages} {220601} (\bibinfo {year}
  {2007})}\BibitemShut {NoStop}%
\bibitem [{\citenamefont {Haller}\ \emph {et~al.}(2015)\citenamefont {Haller},
  \citenamefont {Hudson}, \citenamefont {Kelly}, \citenamefont {Cotta},
  \citenamefont {Peaudecerf}, \citenamefont {Bruce},\ and\ \citenamefont
  {Kuhr}}]{KuhrFermions}%
  \BibitemOpen
  \bibfield  {author} {\bibinfo {author} {\bibfnamefont {E.}~\bibnamefont
  {Haller}}, \bibinfo {author} {\bibfnamefont {J.}~\bibnamefont {Hudson}},
  \bibinfo {author} {\bibfnamefont {A.}~\bibnamefont {Kelly}}, \bibinfo
  {author} {\bibfnamefont {D.~A.}\ \bibnamefont {Cotta}}, \bibinfo {author}
  {\bibfnamefont {B.}~\bibnamefont {Peaudecerf}}, \bibinfo {author}
  {\bibfnamefont {G.~D.}\ \bibnamefont {Bruce}}, \ and\ \bibinfo {author}
  {\bibfnamefont {S.}~\bibnamefont {Kuhr}},\ }\href@noop {} {\bibfield
  {journal} {\bibinfo  {journal} {Nature Phys.}\ }\textbf {\bibinfo {volume}
  {11}},\ \bibinfo {pages} {738} (\bibinfo {year} {2015})}\BibitemShut
  {NoStop}%
\bibitem [{\citenamefont {Lazarides}\ \emph {et~al.}(2014)\citenamefont
  {Lazarides}, \citenamefont {Das},\ and\ \citenamefont
  {Moessner}}]{PeriodicDriven_1}%
  \BibitemOpen
  \bibfield  {author} {\bibinfo {author} {\bibfnamefont {A.}~\bibnamefont
  {Lazarides}}, \bibinfo {author} {\bibfnamefont {A.}~\bibnamefont {Das}}, \
  and\ \bibinfo {author} {\bibfnamefont {R.}~\bibnamefont {Moessner}},\
  }\href@noop {} {\bibfield  {journal} {\bibinfo  {journal} {Phys. Rev. Lett.}\
  }\textbf {\bibinfo {volume} {112}},\ \bibinfo {pages} {150401} (\bibinfo
  {year} {2014})}\BibitemShut {NoStop}%
\bibitem [{\citenamefont {Cramer}\ and\ \citenamefont {Eisert}(2010)}]{NJP}%
  \BibitemOpen
  \bibfield  {author} {\bibinfo {author} {\bibfnamefont {M.}~\bibnamefont
  {Cramer}}\ and\ \bibinfo {author} {\bibfnamefont {J.}~\bibnamefont
  {Eisert}},\ }\href@noop {} {\bibfield  {journal} {\bibinfo  {journal} {New J.
  Phys.}\ }\textbf {\bibinfo {volume} {12}},\ \bibinfo {pages} {055020}
  (\bibinfo {year} {2010})}\BibitemShut {NoStop}%
\bibitem [{\citenamefont {Kitaev}(2006)}]{cond-mat/0506438}%
  \BibitemOpen
  \bibfield  {author} {\bibinfo {author} {\bibfnamefont {A.}~\bibnamefont
  {Kitaev}},\ }\href@noop {} {\bibfield  {journal} {\bibinfo  {journal} {Ann.
  Phys.}\ }\textbf {\bibinfo {volume} {321}},\ \bibinfo {pages} {2} (\bibinfo
  {year} {2006})}\BibitemShut {NoStop}%
\bibitem [{\citenamefont {Lieb}\ \emph {et~al.}(1961)\citenamefont {Lieb},
  \citenamefont {Schultz},\ and\ \citenamefont {Mattis}}]{Lieb_JW}%
  \BibitemOpen
  \bibfield  {author} {\bibinfo {author} {\bibfnamefont {E.}~\bibnamefont
  {Lieb}}, \bibinfo {author} {\bibfnamefont {T.}~\bibnamefont {Schultz}}, \
  and\ \bibinfo {author} {\bibfnamefont {D.}~\bibnamefont {Mattis}},\
  }\href@noop {} {\bibfield  {journal} {\bibinfo  {journal} {Ann. Phys.}\
  }\textbf {\bibinfo {volume} {16}},\ \bibinfo {pages} {407} (\bibinfo {year}
  {1961})}\BibitemShut {NoStop}%
\bibitem [{\citenamefont {Kraus}\ \emph {et~al.}(2009)\citenamefont {Kraus},
  \citenamefont {Wolf}, \citenamefont {Cirac},\ and\ \citenamefont
  {Giedke}}]{Kraus_WicksTheorem}%
  \BibitemOpen
  \bibfield  {author} {\bibinfo {author} {\bibfnamefont {C.~V.}\ \bibnamefont
  {Kraus}}, \bibinfo {author} {\bibfnamefont {M.~M.}\ \bibnamefont {Wolf}},
  \bibinfo {author} {\bibfnamefont {J.~I.}\ \bibnamefont {Cirac}}, \ and\
  \bibinfo {author} {\bibfnamefont {G.}~\bibnamefont {Giedke}},\ }\href@noop {}
  {\bibfield  {journal} {\bibinfo  {journal} {Phys. Rev. A}\ }\textbf {\bibinfo
  {volume} {79}},\ \bibinfo {pages} {012306} (\bibinfo {year}
  {2009})}\BibitemShut {NoStop}%
\bibitem [{\citenamefont {van Lint}\ and\ \citenamefont {Wilson}(1992)}]{lint}%
  \BibitemOpen
  \bibfield  {author} {\bibinfo {author} {\bibfnamefont {J.~H.}\ \bibnamefont
  {van Lint}}\ and\ \bibinfo {author} {\bibfnamefont {R.~M.}\ \bibnamefont
  {Wilson}},\ }\href@noop {} {\emph {\bibinfo {title} {A course in
  combinatorics}}}\ (\bibinfo  {publisher} {Cambridge University Press},\
  \bibinfo {year} {1992})\BibitemShut {NoStop}%
\end{thebibliography}
\end{document}